\documentclass[11pt]{article}

\usepackage[margin=1.2in]{geometry}
\usepackage{amsmath,amssymb,amsthm}
\usepackage{enumitem}
\usepackage[colorlinks=true,linkcolor=blue,citecolor=blue,urlcolor=blue]{hyperref}

\theoremstyle{plain}
\newtheorem{theorem}{Theorem}[section]
\newtheorem{lemma}[theorem]{Lemma}

\newtheorem{corollary}[theorem]{Corollary}
\theoremstyle{definition}
\newtheorem{definition}[theorem]{Definition}
\newtheorem{remark}[theorem]{Remark}
\newtheorem{example}[theorem]{Example}
\newtheorem{question}[theorem]{Question}

\newcommand{\X}{\mathcal{X}}
\newcommand{\Y}{\mathcal{Y}}
\newcommand{\A}{\mathcal{A}}
\newcommand{\B}{\mathcal{B}}
\newcommand{\add}{\mathsf{add}}
\newcommand{\del}{\mathsf{del}}
\newcommand{\Data}{\mathrm{Data}}
\newcommand{\Live}{\mathrm{Live}}
\newcommand{\eps}{\varepsilon}

\newcommand{\remove}[1]{}

\usepackage{cancel}
\usepackage{xcolor}

\title{Machine Unlearning as Private Retroactive Algorithms}
\author{
Haim Kaplan\textsuperscript{*,$\dagger$} 
\and 
Refael Kohen\textsuperscript{*} 
\and 
Yishay Mansour\textsuperscript{*,$\dagger$} 
\and 
Kobbi Nissim\textsuperscript{$\ddagger$,$\dagger$} 
\and 
Uri Stemmer\textsuperscript{*,$\dagger$}
}
\date{}

\begin{document}
\maketitle

\begingroup\renewcommand{\thefootnote}{} \footnotetext{ \hspace{-7px} \textsuperscript{*}Tel Aviv University \quad \textsuperscript{$\dagger$}Google Research \quad \textsuperscript{$\ddagger$}Georgetown University}\endgroup

\begin{abstract}
Machine unlearning typically aims to emulate retraining from scratch: upon a deletion request, the unlearning algorithm should produce an outcome that would have been obtained had the deleted point never been included. 
Recent work has shown that this emulation requirement carries no meaningful privacy semantics against an adversary who observes a sequence of releases. Machine unlearning is thus not a privacy question per se, but rather a data maintenance question, which is precisely the subject of \emph{retroactive algorithms}. 
These are algorithms supporting modifications of past operations, guaranteeing that all subsequent answers reflect the revised history as if it had always been in force. 

We put forward a definition of \emph{private retroactive algorithms}, combining the retroactivity requirement with differential privacy under continual observation. We present constructions achieving both privacy and retroactivity at no asymptotic cost over privacy alone for linear statistics, clustering, and histograms, alongside impossibility results.
\end{abstract}

\section{Introduction}\label{sec:intro}

Machine unlearning addresses the challenge of removing the influence of specific training data from an already trained model. Cao and Yang \cite{CY15} defined a removal as successful if the resulting model's distribution is identical to that of a model retrained from scratch without the removed data (so-called \emph{perfect retraining}). This ensures that the updated model harbors no statistical traces of the removed data, simulating a reality in which it was never part of the training set to begin with.

The primary motivation for machine unlearning stems from modern data privacy laws, such as the GDPR's ``right to be forgotten''. However, recent works \cite{ChenZWBHZ21, CarliniJZPTT22, ChenLZH25, CKNS26, FuNZHW26} have shown that the privacy semantics of this simulation requirement are unclear. First, deleting points may potentially reveal {\em more} information about them via a differencing attack. Second, and perhaps more surprisingly, removing points can inadvertently degrade the privacy of the {\em remaining} data. Specifically, Cohen et al.~\cite{CKNS26} showed that there is a task solvable with differential privacy without deletions (so the task itself is ``benign'' from a privacy standpoint), yet executing just a few deletion requests unavoidably allows an adversary to reconstruct almost all of the remaining data. 

The reason is definitional: perfect retraining constrains what the outputs \emph{are}, whereas a privacy definition must constrain what the outputs \emph{reveal}. 
We therefore propose to build the area on two foundations: a \emph{consistency} guarantee (the answers coherently track the surviving data) and a \emph{leakage} guarantee (the release sequence reveals little about any individual). %

Stripped of its privacy motivation and reduced to its consistency requirement, machine unlearning is, fundamentally, a maintenance question. %
This question predates unlearning, and was studied in the context of dynamic algorithms. Specifically, when deleting a data point requires removing an operation from the middle of an order-sensitive history, the problem falls into the realm of {\em retroactive} data structures (in the sense of Demaine et al.~\cite{DIL07}): modify the past, query the present.

\subsection{Our contributions}\label{subsec:contributions}
Our first contribution is to formalize the notion of \emph{private retroactive algorithms}, uniting the requirements of retroactive data structures \cite{DIL07} and differential privacy \cite{DMNS06}. Crucially, retroactivity (and by extension, machine unlearning) operates in an inherently continual setting. Even a single deletion exposes the system's state both before and after the update; therefore, a meaningful privacy guarantee must bound the cumulative leakage across this evolving process, as in the continual observation model of \cite{DNPR10}. 
Informally, a \emph{private retroactive} algorithm $\A$ is a dynamic algorithm that processes a sequence of data insertions and deletions while satisfying the following:
\begin{itemize}
    \item \textbf{Privacy.} Changing or removing a single update at any point in the execution has almost no effect on the joint distribution of the entire output sequence, as in \cite{DNPR10}.
    \item \textbf{Retroactivity.} For any two input sequences that induce the same surviving data from some timestep $i$ onward (even if they differ by many updates in the prefix), the joint distributions of their answer suffixes from time $i$ coincide, as in \cite{DIL07}. We also consider a weaker version of \emph{marginal approximate retroactivity} in which we only require
    the distribution of the answers for each time step $j\ge i$ to approximately coincide. 
\end{itemize}

These two requirements are of very different natures. Ignoring computational costs, retroactivity on its own is always achievable, with no error: simply recompute, releasing at every time step a fresh answer drawn as a function of the current data alone. Privacy, on the other hand, is highly restrictive even in isolation: it necessitates approximation errors; for instance, privately maintaining a counter requires error $\Omega(\log^{3/2} T)$ \cite{DNPR10,HenzingerUU23,CohenNSS24,BairaktariLarsen26}. This asymmetry suggests measuring the cost of the combination against the cost of privacy alone, and we pose the following question.

\begin{question}\label{q:free}
Is retroactivity free? That is, suppose that a task can be solved with error $\Delta$ by a continually private algorithm. Can it always be solved with error $O(\Delta)$, or comparable, by an algorithm that is in addition retroactive? Or is there a task that admits a private algorithm with small error, but for which every private \emph{and} retroactive algorithm must incur a significantly larger error?
\end{question}

In this work, we demonstrate that there are many interesting cases where exact retroactivity is indeed free. Furthermore, we establish that weaker forms of retroactivity are often attainable generically. Conversely, retroactivity is not universally free: we prove that in certain settings, it strictly forbids instance-adaptive error. We next provide an informal survey of these results.

\paragraph{Exact retroactivity for oblivious-noise mechanisms.} We begin with the class of functionalities where retroactivity is entirely free: real-valued linear statistics over the surviving data. As \cite{BDKT12} showed, DP algorithms for many such functionalities can w.l.o.g.\ be transformed into {\em oblivious-noise} algorithms, i.e., with a noise distribution that is data-independent. This independence is perfectly suited for retroactivity, as it guarantees the noise itself carries no hidden traces of the erased history. Consequently, we demonstrate that the classical tree mechanism of Dwork et al.\ \cite{DNPR10} naturally satisfies exact retroactivity with no asymptotic cost over privacy. As an application, we observe that the recent continually-DP algorithms for $k$-means and $k$-median clustering of \cite{DHS23} are exactly retroactive, essentially as is. This shows that our definition is not only attainable for free in certain settings, but is in fact already satisfied by a large family of existing algorithms in the literature.

\begin{theorem}[Informal version of Theorem~\ref{thm:linear}]
Any one-shot differentially private mechanism for linear queries that relies on data-independent (oblivious) noise can be transformed into a dynamic algorithm that is exactly retroactive and differentially private under continual observation. This transformation incurs no asymptotic overhead in error compared to the standard non-retroactive transformation from the one-shot to the continual DP setting.
\end{theorem}

\paragraph{Stability-based histograms with non-oblivious noise.} We next tackle the task of {\em sparse histograms} over a huge universe, for which oblivious-noise mechanisms do not exist. The standard differential privacy remedy (adding noise only to non-empty bins and thresholding \cite{KKMN09}) is highly data-dependent, seemingly violating the retroactivity requirement. We overcome this by dynamically maintaining a noisy histogram with oblivious noise, and crucially delaying the non-linear thresholding step to the release moment. We show that this achieves exact retroactivity while maintaining an $\ell_\infty$ error independent of the universe size.

\begin{theorem}[Informal version of Theorem~\ref{thm:histogram}]
For any data universe $\mathcal{X}$, there exists an exactly retroactive and continually differentially private algorithm for maintaining histograms. At every time step, the algorithm outputs a sparse histogram with an $\ell_\infty$ error that scales only polylogarithmically with the time horizon $T$, and is completely independent of the universe size $|\mathcal{X}|$.
\end{theorem}

\paragraph{Generic compilers.} Beyond specific functionalities, we present two generic compilers turning DP algorithms into retroactive ones via black box transformations. The first transformation achieves marginal approximate retroactivity while keeping the error essentially the same as that of the base DP algorithm. The second transformation achieves exact retroactivity at the price of specializing to sliding-window functionalities and increasing the error.

\begin{theorem}[Informal version of Theorems~\ref{thm:compiler} and~\ref{thm:sampling}]
Differentially private algorithms can be generically compiled to achieve retroactivity in two ways:
\begin{itemize}
    \item \textbf{Marginal approximate retroactivity:} Any continually differentially private algorithm with real valued answers can be transformed to satisfy marginal approximate retroactivity. The transformation increases the error by an additive term proportional to the algorithm's baseline accuracy.
    \item \textbf{Exact retroactivity for sliding windows:} Any one-shot differentially private algorithm can be compiled into a continually private and exactly retroactive algorithm for sliding-window functionalities, leveraging random sampling of the active window. The privacy cost of the resulting algorithm depends on the size of the window (a large window increases the privacy cost via composition) and the size of the sample (a small sample from a large window can amplify privacy).
\end{itemize}
\end{theorem}

\paragraph{The price of retroactivity.} Finally, we illustrate a negative side of Question \ref{q:free} by showing that retroactivity is not universally free. Specifically, prior work on the Count Distinct problem \cite{JKRSS23} showed that continually-DP algorithms can achieve instance-dependent error, with lower error on ``easy'' instances. We prove that retroactivity strictly forbids such instance-adaptive error guarantees, and that any private and retroactive algorithm must incur worst-case error even on the easiest instances.

\begin{theorem}[Informal version of Theorem~\ref{thm:countdistinct}]
For the Count Distinct problem, any algorithm that is both continually differentially private and retroactive must incur a worst-case polynomial error even on the easiest, insertion-only sequences. Consequently, retroactivity strictly precludes the instance-adaptive, polylogarithmic error guarantees that are achievable by algorithms satisfying privacy alone.
\end{theorem}

\subsection{Related work}\label{subsec:related}

The concept of machine unlearning was pioneered by Cao and Yang \cite{CY15}, who formalized the goal of unlearning as updating a model so that its output distribution perfectly matches that of a model trained from scratch without the deleted data. Following this, a rich line of work has proposed efficient algorithms to achieve this ``perfect retraining'' ideal (or statistical approximations of it) without incurring the full computational cost of retraining. These approaches can broadly be divided into two categories. First, \emph{exact unlearning} methods guarantee an outcome identically distributed to retraining from scratch, often by strategically partitioning the training data so that a deletion only requires retraining a small, isolated fraction of the model \cite{Bourtoule21}. Second, a large body of work focuses on \emph{approximate unlearning}, where the updated model is only required to be statistically close to the retrained model. This relaxation allows for much faster update procedures, typically by leveraging gradient steps, influence functions, or bounded statistical approximations \cite{Ginart19, GGHV20, Izzo21, NRS21, SAKS21}. In addition, several works have utilized techniques from the DP literature {\em as a tool} to achieve the approximate-retraining requirement, or to maintain the retraining guarantee across sequential and adaptive deletion requests \cite{Ullah21, GJNRSW21}. In contrast, in this work privacy is not a tool but a requirement of its own, imposed jointly with retroactivity on the entire release sequence.

Retroactive data structures were introduced by Demaine, Iacono, and Langerman \cite{DIL07}.
Their motivation was to allow inserting or deleting an operation in the past, while maintaining at any time the data structure that would have been obtained by the current sequence of operations. Our definition is an adaptation of their definition to our setting. 
A related notion of \emph{history independent}
data structures was introduced by Micciancio \cite{Mic97} and formalized by Naor and Teague \cite{NT01}. This definition is stronger and requires the \emph{memory representation} of the data structure to be blind also to the order in which the surviving elements were inserted.

\section{The model and the definitions}\label{sec:model}

Let $\X$ be a data domain and $T\in\mathbb{N}$ a time horizon. An \emph{input sequence} is a vector $S=(u_1,\dots,u_T)$ where each $u_i$ is a (possibly empty) multiset of updates arriving at time step $i$: an update is either an addition $(\add,x)$ of a point $x\in\X$ or a deletion $(\del,t,x)$, which, as in \cite{DIL07}, names the addition it cancels via its time stamp $t$. 
Input sequences are arbitrary: a deletion that names no existing addition (or one already outside its scope) simply has no effect. The \emph{current data} at time $i$ is the multiset $\Data_i(S)$ of surviving time-stamped additions. Formally, $\Data_i(S)$ can be defined using the following process, starting from the empty multiset and processing $u_1,\dots,u_i$ in order. For every $u_t$: first, for every $(\add,x)\in u_t$, add a copy of $(t,x)$ to $\Data_i(S)$; then, for every $(\del,t',x)\in u_t$, if $(t',x)\in\Data_i(S)$ then delete one copy of $(t',x)$ from $\Data_i(S)$ (and otherwise do nothing). We call an element of $\Data_i(S)$ a \emph{stamped point}: a pair $(t,x)$ whose \emph{stamp} $t$ is its arrival round, counted with multiplicity. Two additions of the same value at the same round contribute two copies of the same stamped point, and a deletion removes one copy.

\begin{example}\label{ex:running}
The following sequence $S$, with horizon $T=6$ and two values $a,b\in\X$, exercises the conventions above and will illustrate the definitions below.
\begin{center}\small
\begin{tabular}{c l l l}
$\ell$ & $u_\ell$ & $\Data_\ell(S)$ & \\
\hline
$1$ & $\{(\add,a)\}$ & $\{(1,a)\}$ & \\
$2$ & $\{(\add,b),(\add,b)\}$ & $\{(1,a),(2,b),(2,b)\}$ & the pair $(2,b)$ has multiplicity two\\
$3$ & $\{(\del,2,b)\}$ & $\{(1,a),(2,b)\}$ & one copy removed\\
$4$ & $\{(\del,1,b)\}$ & $\{(1,a),(2,b)\}$ & names a pair never added: no-op\\
$5$ & $\{(\del,2,b)\}$ & $\{(1,a)\}$ & multiplicity exhausted\\
$6$ & $\{(\del,2,b)\}$ & $\{(1,a)\}$ & exhausted: no-op\\
\end{tabular}
\end{center}
\end{example}

\begin{remark}\label{rem:model}
For simplicity, one may imagine that a deletion $(\del,t,x)$ arrives only at time steps $t'>t$, after the addition it names, though this is not part of the formal definitions.\footnote{ The process defining $\Data_i(S)$ is well-defined on every input sequence. An addition deleted within its own time step never appears in any $\Data_i(S)$, since within each $u_t$ additions are processed before deletions; and a deletion naming a future addition is a permanent no-op, since the named point is not yet present when the deletion is processed and updates are never revisited.} Time stamps induce a partial order on $\Data_i(S)$; the relative order of additions sharing a time step is unspecified, and the algorithm is free to choose it. Appendix~\ref{app:insertions} discusses the extension in which additions, too, may reach into the past.
\end{remark}

A \emph{dynamic algorithm} $\A$ is a (possibly stateful, randomized) algorithm that at every time step $i\in[T]$ receives $u_i$ and outputs an answer $a_i\in\Y$. We write $\A(S)=(a_1,\dots,a_T)$, a distribution over $\Y^T$, and $[\A(S)]_i^{j}=(a_i,\dots,a_j)$.

\subsection{Privacy} The unit of protection is a single operation, an addition or a deletion, as formally defined below. 

\begin{definition}[Neighboring sequences]\label{def:neighboring}
Input sequences $S,S'$ are \emph{neighboring} if one of them can be obtained from the other by inserting a single update into some $u_t$: either an addition $(\add,x)$, or a deletion $(\del,t',x)$.
\end{definition}

\begin{remark}[Lifecycle protection]\label{rem:lifecycle}
One point's full \emph{lifecycle}, an addition together with the deletion that later cancels it, differs by two operations, and is thus protected by group privacy.
\end{remark}

\begin{definition}[Indistinguishability]\label{def:indist}
Random variables $X,Y$ taking values in a common measurable space are \emph{$(\eps,\delta)$-indistinguishable}, denoted $X\approx_{(\eps,\delta)}Y$, if for every event $B$,
\[
\Pr[X\in B]\;\le\;e^{\eps}\cdot\Pr[Y\in B]+\delta
\qquad\text{and}\qquad
\Pr[Y\in B]\;\le\;e^{\eps}\cdot\Pr[X\in B]+\delta.
\]
\end{definition}

\begin{definition}[$(\eps,\delta)$-privacy \cite{DMNS06,DNPR10}]\label{def:privacy}
A dynamic algorithm $\A$ is \emph{$(\eps,\delta)$-private} if for every pair of neighboring input sequences $S,S'$,
\[
[\A(S)]_1^T\;\approx_{(\eps,\delta)}\;[\A(S')]_1^T.
\]
\end{definition}

Definition~\ref{def:privacy} quantifies over fixed pairs of sequences, guaranteeing privacy w.r.t.\ oblivious input streams. Appendix~\ref{sec:adaptive} defines the adaptive strengthening, in which the stream is chosen as a function of past answers. 

\subsection{Retroactivity}
The second requirement is the distributional core of perfect retraining: histories that agree on the surviving data from some time onward must induce identically distributed answers from that time onward. 

\begin{definition}[Uniting sequences]\label{def:uniting}
Input sequences $S,S'$ \emph{unite at time $i\in[T]$} if $\Data_{\ell}(S)=\Data_{\ell}(S')$ as multisets, for every $i\le\ell\le T$.
\end{definition}

Note that $S,S'$ could potentially differ by {\em many} additions/deletions, and can thus be very far from being neighboring as in Definition~\ref{def:neighboring}.

\begin{definition}[Retroactivity, adapted from \cite{DIL07}]\label{def:retroactivity}
A dynamic algorithm $\A$ is \emph{retroactive} if for every pair of input sequences $S,S'$ that unite at time $i$, 
\[
[\A(S)]_i^T\;\equiv\;[\A(S')]_i^T
\]
i.e., the two blocks of answers are identically distributed.
\end{definition}

In Example~\ref{ex:running}, the sequence $S''$ with $u''_1=\{(\add,a)\}$ and all other rounds empty satisfies $\Data_\ell(S'')=\Data_\ell(S)$ for $\ell\ge5$ and for no earlier $\ell$, so $S$ and $S''$ unite at $5$: retroactivity requires $(a_5,a_6)$ to be identically distributed under $S$ and $S''$, so the answers from round $5$ onward may not reveal that $b$ ever existed. With counting as the functionality, the true answer sequences are $(1,3,2,2,1,1)$ under $S$ and $(1,1,1,1,1,1)$ under $S''$, agreeing exactly on the united suffix.

\begin{definition}[Approximate retroactivity]\label{def:approx-retro}
Let $\alpha,\beta\ge0$. A dynamic algorithm $\A$ is \emph{$(\alpha,\beta)$-approximately retroactive} if for every pair of input sequences $S,S'$ that unite at time $i$,
\[
[\A(S)]_i^T\;\approx_{(\alpha,\beta)}\;[\A(S')]_i^T,
\]
and it is \emph{$(\alpha,\beta)$-marginally approximately retroactive} if for every such pair and every $i\le\ell\le T$,
\[
[\A(S)]_\ell^{\,\ell}\;\approx_{(\alpha,\beta)}\;[\A(S')]_\ell^{\,\ell}.
\]
\end{definition}

Retroactivity is exactly $(0,0)$-approximate retroactivity, and $(\alpha,\beta)$-approximate retroactivity implies its marginal counterpart, since every single answer is a marginal of the suffix block. The converse direction fails: the marginal requirement constrains each answer in isolation and places no constraint on how the answers may jointly encode the pre-uniting history. Section~\ref{sec:compiler} quantifies the gap.

\subsection{Private retroactive algorithms}

Our central definition combines the two requirements as follows.

\begin{definition}[Private retroactive algorithm]\label{def:main}
A dynamic algorithm $\A$ is an \emph{$(\eps,\delta)$-private retroactive algorithm} if it satisfies Definitions~\ref{def:privacy} and~\ref{def:retroactivity}.
\end{definition}

Note that Definition~\ref{def:main} is trivially satisfied by an algorithm that always outputs $\bot$. Accuracy with respect to a functionality of interest is imposed separately, and the subject of this paper is the trade off between privacy, retroactivity, and accuracy.

\section{Linear functionalities with oblivious noise}\label{sec:linear}

We begin with the class of functionalities for which the answer to Question~\ref{q:free} is affirmative in the most direct way: real valued linear statistics of the surviving data. The construction is the classical tree mechanism of \cite{DNPR10}, evaluated on the current data and carrying persistent, data independent additive noise, and the point of this section is that it is exactly retroactive as it stands, at no asymptotic cost over privacy under continual observation. 

Throughout, $h(D)\in\mathbb{Z}_{\ge0}^{\X}$ denotes the histogram of the {\em values} of a multiset $D\subseteq[T]\times\X$ of stamped points (ignoring the stamps). That is, $h(D)_x$ is the number of stamped points in $D$ with value $x$. We also write $F\in\mathbb{R}^{d\times|\X|}$ to denote a query matrix with columns $(F_x)_{x\in\X}$, where $d$ denotes the number of queries. The functionality of interest maps the current data to $Fh(\Data_\ell(S))\in\mathbb{R}^d$. 

\begin{definition}[Oblivious-noise mechanism \cite{BDKT12}]\label{def:oblivious}
A one shot mechanism $M$ for the query $F$ is an \emph{oblivious-noise mechanism} if $M(D)\equiv Fh(D)+N$, where $N\sim\nu$ for a distribution $\nu$ on $\mathbb{R}^d$ that does not depend on the input.
\end{definition}

We will show that any such (one-shot) oblivious-noise mechanism can be transformed into a retroactive one. As Bhaskara et al.~\cite{BDKT12} showed, many mechanisms for answering linear queries can be turned into oblivious-noise mechanisms, and so this construction captures quite a few settings.

\paragraph{The construction.}
Consider a complete binary tree whose leaves correspond to single time steps in $[T]$. Every node in this tree corresponds to the time interval spanned by its children (these are dyadic intervals). For $\ell\in[T]$ let $D(\ell)$ denote the canonical decomposition of $[1,\ell]$ into at most $m=\log_2(T)+1$ such dyadic intervals (at most one node/interval from every level of the tree), and recall that for every $t\le\ell$ exactly one interval of $D(\ell)$ contains $t$. 

Let $M$ be an oblivious-noise mechanism for a query matrix $F$ with noise distribution $\nu$. 
Define $\A_{F,\nu}$ to be the algorithm that maintains one persistent noise vector $N_v\sim\nu$ per node $v$ of the tree, mutually independent. At every round $\ell$, algorithm $\A_{F,\nu}$ releases
\begin{equation}
R_\ell\;=\;Fh(\Data_\ell(S))\;+\;\sum_{v\in D(\ell)}N_v\,. \label{eq:1}
\end{equation}

\begin{theorem}\label{thm:linear}
Algorithm $\A_{F,\nu}$ satisfies:
\begin{enumerate}
\item $\A_{F,\nu}$ is retroactive (as in Definition~\ref{def:retroactivity}).
\item If $M$ is $(\eps_0,\delta_0)$-DP then $\A_{F,\nu}$ is $(2m\eps_0,\,2m\delta_0)$-private (as in Definition~\ref{def:privacy}), where $m=\log_2 T+1$.
\item At every round $\ell$, the error $R_\ell-Fh(\Data_\ell(S))$ is the sum of at most $m$ independent samples from $\nu$.
\end{enumerate}
\end{theorem}

\begin{proof}
\emph{Retroactivity.} Fix $S,S'$ uniting at time $i$. The answer block $[\A(S)]_i^{\,T}=(R_i,\dots,R_T)$, with $R_\ell$ the release defined in the construction above, is a randomized function of $\{\Data_\ell(S)\}_{\ell=i}^T$, and is otherwise independent of $S$. As this is identical to $\{\Data_\ell(S')\}_{\ell=i}^T$, we have that $[\A(S)]_i^{\,T}\equiv[\A(S')]_i^{\,T}$.

\emph{Privacy.} First recall that as $M$ is noise-oblivious $(\eps_0,\delta_0)$-DP, then for any element $x\in\X$ it holds that $\nu\approx_{(\eps_0,\delta_0)}\nu\pm F_x$, since for two datasets differing by adding/removing $x$ we have that $Fh$ differs by the column $F_x$, and $M$ must hide this difference.

Now let $S,S'$ be neighboring input sequences. There is a contiguous (possibly empty) time window $W\subseteq[T]$ throughout which the histograms $h(\Data_\ell(S))$ and $h(\Data_\ell(S'))$ differ by the addition/removal of one fixed element $x^*$; these histograms are otherwise equal throughout the execution.\footnote{Indeed, the differing operation changes the multiplicity of a single stamped point by one, and once the two multiplicities coincide they remain equal, so the differing rounds form an interval.} For simplicity let us assume that the window $W$ ends at time $T$, so that we only need to argue about the starting time $w_s$ of the window.

Now observe that the time step $w_s$ participates in exactly $m$ dyadic intervals/nodes, and denote these intervals as $L$. Observe that in any time step $\ell\geq w_s$, exactly one node from $L$ participates in $D(\ell)$ (and thus its corresponding noise $N_v$ participates in the output, see Equation~(\ref{eq:1})). Hence, shifting every noise $\{N_v\}_{v\in L}$ by $\pm F_{x^*}$ makes the outputs of the two executions identical at every round. These are $m$ coordinates of a product measure, each shift $(\eps_0,\delta_0)$-indistinguishable, so by composition the transcripts are $(m\eps_0,m\delta_0)$-indistinguishable.\footnote{A general window is handled by additionally shifting, in the opposite direction, the $m$ nodes containing the first time step after the window (a node containing both endpoints is left unchanged), doubling the parameters to $(2m\eps_0,\,2m\delta_0)$.}

\emph{Utility.} Immediate from Equation~(\ref{eq:1}): $|D(\ell)|\le m$ and the noises $N_v$ are mutually independent.
\end{proof}

For $d=1$ and $F$ the all ones row, instantiating $\nu$ with Laplace noise recovers the standard counter of \cite{DNPR10}, with error $\tilde{\Theta}(\log^{3/2}T)$, matching the non-retroactive lower bound for private counting \cite{DNPR10,HenzingerUU23,CohenNSS24,BairaktariLarsen26}. As the construction is exactly the standard non retroactive one, retroactivity is free for linear queries. This is the affirmative side of Question~\ref{q:free} for the linear class.

\subsection{A case study: clustering under continual observation}\label{subsec:clustering}

In this section we show that the construction of Theorem~\ref{thm:linear} extends beyond linear queries. Specifically, we revisit a result by Dupr\'e la Tour, Henzinger, and Saulpic \cite{DHS23} who presented algorithms for $k$-means and $k$-median clustering under continual observation, over streams of insertions and deletions of points from a bounded ball in $\mathbb{R}^d$. We observe that their construction is exactly retroactive as it stands, satisfying our Definition~\ref{def:main}.

\paragraph{The model.}
Fix a dimension $d$, a diameter $\Lambda>0$, and let the value domain be $B(0,\Lambda)=\{p\in\mathbb{R}^d:\|p\|_2\le\Lambda\}$. We consider a setting with at most one update per round: each update $u_i$ is either empty, or a single addition $(\add,p)$ of a point $p\in B(0,\Lambda)$, or a single deletion $(\del,t,p)$ naming a previously added stamped point. As always, $\Data_\ell(S)$ is the multiset of surviving stamped points, and multiplicities are allowed. The functionality releases at every round $\ell$ a set of $k$ centers $c_1,\dots,c_k\in\mathbb{R}^d$, evaluated by the $k$-means cost
\[
\mathrm{cost}_\ell(c_1,\dots,c_k)\;=\;\sum_{(t,p)\in\Data_\ell(S)}\ \min_{1\le j\le k}\|p-c_j\|_2^2\,,
\]
and $\mathrm{OPT}_\ell$ denotes the minimum of this cost over all sets of $k$ centers. Neighboring sequences are as in Definition~\ref{def:neighboring}, with the inserted operation occupying a previously empty round.

In our proof, we will use the following closure to post processing property of retroactive algorithms.

\begin{lemma}[Closure under memoryless post processing]\label{lem:closure}
Let $\A$ be a retroactive dynamic algorithm, let $\rho$ be a random variable drawn once from a fixed, data independent distribution and independently of everything else, and let $\B$ be the algorithm that at every round $\ell$ releases $b_\ell=g_\ell(a_\ell;\rho,\xi_\ell)$, where $a_\ell$ is the round $\ell$ answer of $\A$, the maps $g_\ell$ are fixed, and the coins $\xi_\ell$ are fresh and independent across rounds. Then $\B$ is retroactive. %
\end{lemma}

Note that $g_\ell$ consumes only the \emph{current} answer $a_\ell$. Post processing that also consults earlier answers does not preserve retroactivity in general, as it may correlate the answer block with the pre uniting past.

\begin{proof}
Fix $S,S'$ uniting at $i$. The block $(b_i,\dots,b_T)$ is a deterministic function of the block $(a_i,\dots,a_T)$, of $\rho$, and of the fresh coins $(\xi_i,\dots,\xi_T)$. The three are mutually independent, the distribution of $(a_i,\dots,a_T)$ is the same under $S$ and $S'$ by the retroactivity of $\A$, and the distributions of $\rho$ and of the coins are data independent. Hence the joint distribution of the inputs to the common deterministic function is the same under $S$ and $S'$, and so is the distribution of the block.
\end{proof}

\begin{theorem}\label{thm:kmeans}
In the $k$-means model above, for every $\alpha,\eps>0$ there is an algorithm that is $(\eps,0)$-private (Definition~\ref{def:privacy}) and retroactive (Definition~\ref{def:retroactivity}), hence a private retroactive algorithm (Definition~\ref{def:main}), which at every round $\ell$ releases a set of $k$ centers, and which satisfies, with probability at least $0.99$, simultaneously for all rounds $\ell$: the $k$-means cost of the released centers on $\Data_\ell(S)$ is at most
\[
(1+\alpha)w^*\cdot\mathrm{OPT}_\ell\;+\;k^{O_\alpha(1)}\,d^2\log(n)^4\log(T)^{3.001}\cdot\Lambda^2/\eps,
\]
where $n=\max_\ell |\Data_\ell(S)|$ and $w^*$ is the best approximation ratio of non private static $k$-means.
\end{theorem}

\begin{proof}
The algorithm is that of \cite{DHS23}. We begin with a simplified description of their algorithm.

\begin{center}
{\setlength{\fboxsep}{10pt}%
\fbox{\begin{minipage}{0.88\textwidth}
\textbf{The algorithm of \cite{DHS23}, simplified.}
The algorithm runs in parallel $\lfloor\log T\rfloor$ \emph{instances} of a greedy clustering procedure, each with privacy budget $\eps/\lfloor\log T\rfloor$. Each instance succeeds at any given round with constant probability only; running many boosts this probability, and at every round the algorithm releases the answer of the instance that currently looks best.
\begin{itemize}
\item \emph{Setup of an instance (data independent).} First, apply a random projection $\pi_i$ to reduce the $d$-dimensional space down to a smaller $O(\log k)$-dimensional space. Then, rigidly carve this projected space into a fixed, multi-layered grid of sub-regions (a net decomposition). Crucially for retroactivity, this entire spatial structure is generated completely blindly, without ever looking at the data.
\item \emph{Counters of an instance.} For every region, three running counters are maintained via the binary mechanism: the number of surviving points in the region, the sum of their original $d$ dimensional coordinates, and the sum of their squared norms. Every arriving update is fed to all instances.
\item \emph{Releases (per round).} Each instance computes $k$ candidate centers and an estimate of their cost by post-processing its current noisy region counts, current noisy sums, and current noisy norms (without re-accessing the data). The algorithm releases the solution of the instance with the smallest estimated cost.
\end{itemize}
\end{minipage}}}
\end{center}

Our proof uses only two properties of their algorithm, both visible in the description above and verified by inspection of their construction.

\begin{enumerate}
\item[(P1)] For each instance $i$ there is a matrix $F_i$, determined by the instance's random projection $\pi_i$ and by the fixed family of regions, such that the instance's counter vector at round $\ell$ equals $F_i\,h(\Data_\ell(S))+\sum_{v\in D(\ell)}N^{(i)}_v$, where $D(\ell)$ is the dyadic decomposition of $[1,\ell]$ as in Section~\ref{sec:linear} and the $N^{(i)}_v$ are the persistent independent noises of the instance's binary mechanisms.\footnote{The counters are the region counts, the per region coordinate sums, and the per region sums of squared norms; each point lies in $k^{O(1)}\log n$ regions, so the columns of $F_i$ are sparse. Ineffective deletions, which our totality convention permits, are discarded on arrival and reach no counter.}
\item[(P2)] The releases at every round are post processing of the current counter vectors and the projections, with fresh and independent coins across rounds.
\end{enumerate}

\emph{Retroactivity.} By (P1), the joint counter block of all instances from any round $i$ onwards is a randomized function of $\{\Data_\ell(S)\}_{\ell=i}^T$, and is otherwise independent of $S$: the noise index sets $D(\ell)$ are determined by the round alone, and the projections are drawn once from fixed data independent distributions. By (P2) and Lemma~\ref{lem:closure}, with $\rho=(\pi_i)_i$, the same holds for the released centers, exactly as in the retroactivity part of Theorem~\ref{thm:linear}.

\emph{Privacy and accuracy.} These follow directly from the guarantees of \cite{DHS23} for their algorithm, which we run unchanged. (Technically, we need to run their algorithm with privacy parameter $\eps/2$, as they have a slightly different notion for which sequences are neighboring, but this does not change anything.)
\end{proof}

\begin{corollary}[$k$-median]\label{cor:kmedian}
The $k$-median algorithm of \cite{DHS23}, their Theorem~2, run with privacy parameter $\eps/2$, is $(\eps,0)$-private (Definition~\ref{def:privacy}) and retroactive (Definition~\ref{def:retroactivity}), with the accuracy guarantee stated there.
\end{corollary}

\begin{proof}
The algorithm is covered by the same argument, and is in fact simpler: it runs a single instance, with no boosting and no cost estimates, and its releases are again per round post processing of counters maintained over fixed regions. Properties (P1) and (P2) hold verbatim, so retroactivity and privacy follow exactly as in the proof of Theorem~\ref{thm:kmeans}.
\end{proof}

\begin{remark}[Unknown horizon]\label{rem:unknown-horizon}
In our model the horizon $T$ is fixed in advance, and the algorithm above uses it. As in \cite{DHS23}, this dependence can be removed by starting more and more instances as time goes by, the $i$th instance at round $2^i$, initialized with the currently surviving points, at the cost of a slightly larger error. For retroactivity, a minor change to their initialization is then required, so that the noise structure remains independent of the past data.
\end{remark}

\section{Stability-based histograms}\label{sec:histogram}

We now give a construction for the histogram functionality over a huge bin universe $\X$: each addition places an item in a bin, each deletion removes one, and at every round the algorithm releases a sparse approximate histogram of the surviving data, with $\ell_\infty$ error independent of $|\X|$. The bin counts form a linear statistic, so Theorem~\ref{thm:linear} applies; but applied as stated it instantiates a noise variable for every bin of the universe, so its release is dense and its $\ell_\infty$ error grows with $\log|\X|$. The offline remedy is the stability based histogram \cite{KKMN09}: add noise only to the nonempty bins, release those whose noisy count crosses a threshold, and output exact zeros elsewhere. This keeps the histogram sparse (the vast majority of the $|\X|$ bins remain empty) and makes the $\ell_\infty$ error independent of $|\X|$.

The difficulty is that with this remedy the noise is not oblivious anymore: which bins carry noise depends on the data, and, when run dynamically, on the history of activity. Noise of this data dependent kind breaks the retroactivity analysis of Section~\ref{sec:linear}. We overcome this by delaying the nonlinear step of the computation to the release moment: informally, we maintain dynamically a noisy histogram with oblivious noise, and add a non linear post processing step that sparsifies the release before every round.

Throughout this section, updates are exactly as in Section~\ref{sec:model}: arbitrary multisets of additions and deletions, with no restriction. For a bin $b\in\X$ write $c_b(\ell)$ for the number of copies with value $b$ in $\Data_\ell(S)$. We retain the dyadic notation of Section~\ref{sec:linear}.\footnote{Recall: the nodes of the complete binary tree over $[T]$ are the dyadic intervals, $D(\ell)$ is the canonical decomposition of $[1,\ell]$ into at most $m=\log_2(T)+1$ such intervals, at most one per level, and every $t\le\ell$ lies in exactly one interval of $D(\ell)$.}

\paragraph{The construction.}
Let $\eps_0$ be a privacy parameter and let $\tau$ be a threshold parameter. The algorithm $\A_{\eps_0,\tau}$ stores the data, and maintains a table of noise variables $\eta_{v,b}\sim\mathrm{Lap}(1/\eps_0)$, indexed by pairs of a dyadic interval $v$ and a bin $b$, mutually independent, each instantiated the first time it is used and stored thereafter. At every round $\ell$, for each bin $b$ with $c_b(\ell)\ge1$ it computes
\[
V_\ell(b)\;=\;c_b(\ell)\;+\;\sum_{v\in D(\ell)}\eta_{v,b}\,,
\]
and it releases the sparse histogram
\[
a_\ell\;=\;\big\{(b,V_\ell(b))\;:\;c_b(\ell)\ge1\ \text{and}\ V_\ell(b)\ge\tau\big\},
\]
interpreted as assigning $0$ to every other bin (so that a bin whose history was erased by deletions is indistinguishable from one never touched). Only finitely many noise variables are ever instantiated: at most $m$ per surviving bin per round, each drawn on first use and reused in all subsequent rounds.

\begin{theorem}\label{thm:histogram}
For every $\eps_0>0$ and $\delta\in(0,1)$, set $\tau=1+\frac{m}{\eps_0}\ln\frac{2mT}{\delta}$. Then $\A_{\eps_0,\tau}$ is retroactive (Definition~\ref{def:retroactivity}), and it is $(\eps,\delta)$-private (Definition~\ref{def:privacy}) with $\eps=2m\eps_0$.
\end{theorem}

\begin{proof}[Proof of retroactivity]
Fix $S,S'$ uniting at time $i$. The answer block $[\A(S)]_i^{\,T}=(a_i,\dots,a_T)$ is a deterministic function of the counts $\{c_b(\ell)\}_{\ell\ge i,\,b\in\X}$ (which are determined by $\{\Data_\ell(S)\}_{\ell=i}^T$), together with the noise variables of intervals involving bins that are non-empty at some time $\ell\geq i$, that is,
\[
H_i\;=\;\big\{\eta_{v,b}\;:\;\text{$v\in D(\ell)$ and $c_b(\ell)\ge1$ for some $\ell\ge i$}\big\}.
\]
Note that the index set of $H_i$ is itself determined by $\{\Data_\ell(S)\}_{\ell=i}^T$ and is the same for $S$ and $S'$. Each $\eta_{v,b}$ is an independent $\mathrm{Lap}(1/\eps_0)$ variable regardless of the round at which it was instantiated, which may precede $i$. The noise variables of pairs $(v,b)$ outside $H_i$, including those of bins whose entire activity was erased by deletions before round $i$, do not influence the answer block, because bins with $c_b(\ell)=0$ are released as exact zeros. So $[\A(S)]_i^{\,T}\equiv[\A(S')]_i^{\,T}$.
\end{proof}

\begin{proof}[Proof of privacy]
Let $S,S'$ be neighboring input sequences. As established in the privacy analysis of Theorem~\ref{thm:linear}, there are a bin $b^*$ and a contiguous (possibly empty) time window $W\subseteq[T]$ throughout which the counts $c_{b^*}(\ell)$ and $c'_{b^*}(\ell)$ of the two executions differ by one unit, all other bins having identical counts at all times. For simplicity let us assume that the window $W$ ends at time $T$, so that we only need to argue about its starting time $w_s$.\footnote{A general window is handled as in Theorem~\ref{thm:linear}, by additionally shifting, in the opposite direction, the $m$ intervals containing the first time step after the window (an interval containing both endpoints is left unchanged); this doubles the number of shifted variables to at most $2m$, and the distortion factor below becomes $e^{2m\eps_0}=e^{\eps}$.}

Now observe that the time step $w_s$ participates in exactly $m$ dyadic intervals, and denote these intervals as $L$; in any time step $\ell\ge w_s$, exactly one interval from $L$ participates in $D(\ell)$. Hence, shifting every noise $\eta_{v,b^*}$, $v\in L$, by $\pm1$, with the sign matching the difference, makes the noisy values $V_\ell(b^*)$ of the two executions identical at every round, all other bins carrying identical values as well. The shift moves $m$ independent $\mathrm{Lap}(1/\eps_0)$ variables by $1$ each, so it distorts probabilities by a factor of at most $e^{m\eps_0}\le e^{\eps}$, in both directions.

Under this coupling the two executions can differ only at rounds $\ell\in W$ where one execution has a true count of 0 and the other has a true count of 1: The execution holding the extra unit releases its value $V_\ell(b^*)$ if it crosses $\tau$ while the other releases an exact zero. By a union bound over the at most $T$ rounds and the $m$ summand noises,
\[
\Pr\left[\exists\ell\in W: 
\begin{array}{l}
\text{in one execution }  V_\ell(b^*)\ge\tau\\
\text{and in the other the true count is } 0
  \end{array}
\right]\;\le\;T\cdot m\cdot e^{-(\tau-1)\eps_0/m}\;\le\;\delta\,.
\]
Overall, for every event $E$ we have $\Pr[\A(S)\in E]\le e^{\eps}\Pr[\A(S')\in E]+\delta$, as required.
\end{proof}

\begin{remark}[Accuracy]\label{rem:histogram-accuracy}
The released histogram never contains a bin with no surviving items, and at bins with $c_b(\ell)\ge1$ its error is at most $|V_\ell(b)-c_b(\ell)|$ plus $\tau$ when the bin is thresholded away. With probability $1-\beta$, simultaneously for all rounds and all released bins, every noise sum is at most $\frac{m}{\eps_0}\ln\frac{mT\cdot A}{\beta}$ in absolute value, where $A$ bounds the number of instantiated pairs; hence the $\ell_\infty$ error is $O\big(\frac{\log^2 T}{\eps}\log\frac{T}{\delta\beta}\big)$ with $\eps=2m\eps_0$, independent of $|\X|$, and with exact zeros outside the true support. We do not optimize the polylogarithmic factors.
\end{remark}

\section{Generic compilers}\label{sec:compiler}

In this section we present two generic constructions that, under certain conditions, allow us to transform differentially private algorithms into retroactive ones.

\subsection{From continual privacy to marginal approximate retroactivity}\label{subsec:smoothing}

In this section, we demonstrate that marginal approximate retroactivity can be achieved generically. Specifically, any accurate algorithm that satisfies continual DP can be transformed into one that additionally satisfies marginal retroactivity. More generally, this transformation guarantees that any block of $B$ consecutive answers remains approximately retroactive, with an approximation error that grows as $\sqrt{B}$. 

\paragraph{Intuition.} Suppose we have a (continual) DP algorithm $\A$, approximating some functionality $f$ of the input sequence. Now recall that once two input sequences unite, then the exact value of $f$ on them becomes identical. Hence, if the algorithm $\A$ is highly accurate, then its outputs for both streams must be very close to this shared truth (say within an error margin $E$).  Noisifying $\A$'s answers proportionally to $E$ masks this difference, thereby making it (marginally) retroactive.

\medskip

For concreteness, throughout this subsection we assume that the target functionality $f$ is a mapping from multisets over $\X$ to $\mathbb{R}^d$, and that the error of the base DP algorithm is measured via the Euclidean distance. See Remark~\ref{rem:compiler-general} for an extension to other metric spaces.

\begin{definition}[Accuracy]\label{def:accuracy}
A dynamic algorithm $\A$ is \emph{$(E,\beta)$-accurate for $f$} if for every input sequence $S$, with probability at least $1-\beta$, $\|a_\ell-f(\Data_\ell(S))\|\le E$ simultaneously for all $\ell\in[T]$.
\end{definition}

Given a dynamic algorithm $\A$ and $\sigma>0$, the \emph{compiled algorithm} $\A^{\sigma}$ runs $\A$, and at every round releases $z_\ell=a_\ell+N_\ell$, where $a_\ell$ is the answer of $\A$ and $N_\ell\sim\mathcal{N}(0,\sigma^2 I_d)$ is fresh and independent across rounds. We use two standard facts about Gaussian noise \cite{DKMMN06,DR14}: for $\alpha\in(0,1]$, $\beta_K\in(0,1)$ and $\sigma=\frac{2E}{\alpha}\sqrt{2\ln(1.25/\beta_K)}$, we have $y+N\approx_{(\alpha,\beta_K)}y'+N$ whenever $\|y-y'\|\le2E$, and $\|N\|\le\sigma\big(\sqrt d+\sqrt{2\ln(1/\beta_K)}\big)=:E_K$ except with probability $\beta_K$.

\begin{theorem}[Compiler]\label{thm:compiler}
Let $\A$ be $(\eps,\delta)$-private (Definition~\ref{def:privacy}) and $(E,\beta_0)$-accurate for $f$, let $\alpha\in(0,1]$ and $\beta_K\in(0,1)$, and let $\sigma$ and $E_K$ be as above. Then:
\begin{enumerate}[itemsep=1pt,topsep=3pt]
\item $\A^{\sigma}$ is $(\eps,\delta)$-private.
\item $\A^{\sigma}$ is $(E+E_K,\,\beta_0+T\beta_K)$-accurate for $f$.
\item For every pair $S,S'$ uniting at time $i$, every block $i\le\ell_1\le\ell_2\le T$ of length $B=\ell_2-\ell_1+1$, and every $\beta'\in(0,1)$,
\[
[\A^{\sigma}(S)]_{\ell_1}^{\ell_2}\;\approx_{(\alpha_B,\,\beta_B)}\;[\A^{\sigma}(S')]_{\ell_1}^{\ell_2},
\]
where $\alpha_B=\min\big\{B\alpha,\ \alpha\sqrt{2B\ln(1/\beta')}+B\alpha(e^{\alpha}-1)\big\}$ and $\beta_B=B\beta_K+\beta'+(2+e^{\alpha_B})\beta_0$.
\end{enumerate}
\end{theorem}

\begin{proof}
\emph{Part 1 (Privacy).} Standard $(\eps,\delta)$-privacy follows directly from the immunity of differential privacy to post-processing, since the compiled transcript is simply the raw transcript with independent Gaussian noise added to every answer.

\emph{Part 2 (Accuracy).} This follows from the triangle inequality, a union bound over the $T$ rounds of the Gaussian norm tail, and the baseline accuracy assumption of $\A$.

\emph{Part 3 (Block Retroactivity).} Fix $S,S'$ uniting at $i$ and a block $[\ell_1,\ell_2]\subseteq[i,T]$ of length $B$. Since the sequences unite at $i$, the true answers $f(\Data_\ell(S))=f(\Data_\ell(S'))$ coincide at every round of the block. Let $G$ be the set of answer vectors $a=(a_{\ell_1},\dots,a_{\ell_2})$ that are $E$-accurate at every round of the block with respect to these true answers. By the accuracy of $\A$, $\Pr[[\A(S)]_{\ell_1}^{\ell_2}\in G]$ and $\Pr[[\A(S')]_{\ell_1}^{\ell_2}\in G]$ are both at least $1-\beta_0$. 
For $a\in G$ let ${\rm GN}(a)$ denote the random vector $(a_\ell+N_\ell)_{\ell=\ell_1}^{\ell_2}$, where the $N_\ell\sim\mathcal{N}(0,\sigma^2 I_d)$ are independent.
Note that any two $a,a'\in G$ are within distance $2E$ at every round of the block. Thus, by the properties of the Gaussian mechanism \cite{DKMMN06,DR14}, composed across the $B$ rounds \cite{DRV10}, for every $a,a'\in G$ we have ${\rm GN}(a)\approx_{(\alpha_B,\beta'_B)}{\rm GN}(a')$, where $\beta'_B=B\beta_K$ under basic composition and $\beta'_B=B\beta_K+\beta'$ under advanced composition. Hence, for any event $F$ we have
\begin{align*}
\Pr\big[[\A^{\sigma}(S)]_{\ell_1}^{\ell_2}\in F\big]
&\le \Pr\big[[\A(S)]_{\ell_1}^{\ell_2}\in G\big]\cdot\sup_{a\in G}\Pr[{\rm GN}(a)\in F]+\beta_0\\
&\le \Big(\Pr\big[[\A(S')]_{\ell_1}^{\ell_2}\in G\big]+\beta_0\Big)\Big(e^{\alpha_B}\inf_{a'\in G}\Pr[{\rm GN}(a')\in F]+\beta'_B\Big)+\beta_0\\
&\le e^{\alpha_B}\Pr\big[[\A^{\sigma}(S')]_{\ell_1}^{\ell_2}\in F\big]+\beta'_B+(2+e^{\alpha_B})\beta_0\,.
\end{align*}
The reverse inequality holds by symmetry.
\end{proof}

\begin{corollary}\label{cor:marginal}
In the setting of Theorem~\ref{thm:compiler}, $\A^{\sigma}$ is $(\alpha,\,\beta_K+(2+e^{\alpha})\beta_0)$-marginally approximately retroactive.
\end{corollary}

\begin{remark}[Beyond $\mathbb{R}^d$]\label{rem:compiler-general}
We worked in $\mathbb{R}^d$ for simplicity, but the argument generalizes to any answer metric space $(\Y,d)$, provided that it admits a randomized map $K$ playing the role of Gaussian noise: $K(y)\approx_{(\alpha,\beta_K)}K(y')$ whenever $d(y,y')\le2E$, and $d(K(y),y)\le E_K$ except with probability $\beta_K$. Nothing else about the noise is used.
\end{remark}

\begin{remark}[An approximately retroactive median with deletions]\label{rem:compiler-median}
As an example of Remark~\ref{rem:compiler-general}, consider the task of reporting the median of the current data (which undergoes insertions and deletions) over a {\em huge} ordered domain. Without retroactivity, this can be solved privately (in the continual-DP model) by maintaining two dyadic trees over $[T]$, one for additions and one for deletions. Each node holds a one shot DP estimate of the empirical CDF of the points processed in its interval, with rank error depending on the domain only through $2^{O(\log^*|\X|)}$ \cite{BNSV15}. At every round the algorithm releases the difference of the two estimates accumulated along $D(\ell)$, which produces an approximate CDF of the remaining data points (from which one could estimate the median by post-processing). This base algorithm handles arbitrary deletions and is private by the usual tree accounting. It is accurate with $E=\mathrm{polylog}(T)\cdot 2^{O(\log^*|\X|)}/\eps$, in the metric of rank discrepancy between CDF estimates. 
Note that we cannot add real-valued noise to the ``non-retroactive'' median estimation, as we aim to approximate the median by rank. In the terminology of Remark~\ref{rem:compiler-general}, the map $K$ which we will use to noisify the non-retroactive CDF is the cumulatively differentially private median algorithm of \cite[Theorem~4.2]{CLNSS23}.\footnote{The non-retroactive estimate is a difference of two CDF estimates; before feeding it to the algorithm of \cite[Theorem~4.2]{CLNSS23} we transform it into a sanitized dataset via standard techniques.} Two CDF estimates within rank discrepancy $2E$ induce inputs at cumulative distance $O(E)$, so running their algorithm with privacy parameter $\Theta(\alpha/E)$ gives, by group privacy, the closeness that Remark~\ref{rem:compiler-general} requires, with rank error $E_K=\widetilde{O}(E\cdot 2^{\log^*|\X|}/\alpha)$. The result is an $(\eps,\delta)$-private median with arbitrary deletions, rank error $\mathrm{polylog}(T)\cdot 2^{O(\log^*|\X|)}/(\eps\alpha)$, and the retroactivity guarantees of Theorem~\ref{thm:compiler}.
\end{remark}

\subsection{Exact retroactivity for window functionalities: sampling with expiry}\label{subsec:sampling}

In this section we present a generic transformation that takes a {\em one shot} DP algorithm $M$ and produces an \emph{exactly} retroactive private algorithm for the sliding window version of $M$'s task. For {\em linear} window statistics, such a result follows directly from Theorem~\ref{thm:linear}, so no transformation is needed. The motivation for this section thus comes from non-linear statistics. Let $W\in[T]$ be a fixed window length. 
The construction is based on maintaining a uniformly random sample from the surviving points in the window, and running $M$ on this sample. Intuitively, the fact that we sample (only) from the surviving points guarantees retroactivity. Privacy will be enforced using the window, which caps every element's influence, amplified by the secrecy of the sample within the window.

\begin{definition}
 A surviving stamped point $(t,x)$ is \emph{live} at round $\ell$ if $\ell-t<W$ (note that $t\leq\ell$ for $(t,x)$ to be surviving at time $\ell$). So, every stamped point expires exactly $W$ rounds after its arrival. Let $\Live_\ell(S)=\{(t,x)\in\Data_\ell(S):\ell-t<W\}$ denote the multiset of live stamped points. %
Write $n_\ell(S)=|\Live_\ell(S)|$.
\end{definition}

\paragraph{The compiler.}
Let $M$ be an algorithm taking a multiset over $\X$ to a distribution over answers, and fix a target sample size $k\ge1$. The compiled algorithm $\A^{\mathrm{samp}}_{M,k,W}$ draws at every round $\ell$, with fresh coins, a uniformly random subset $Q_\ell$ of size $\min\{k,n_\ell\}$ of $\Live_\ell(S)$, and releases $z_\ell=M(Q_\ell)$, applied to the multiset of values of the sampled points, with fresh independent coins for $M$ as well.

\begin{theorem}[Sampling compiler]\label{thm:sampling}
Let $M$ be $(\eps_0,\delta_0)$-differentially private with respect to adding or removing one record. Then the following hold for $\A=\A^{\mathrm{samp}}_{M,k,W}$.
\begin{enumerate}
\item \emph{Retroactivity.} $\A$ is exactly retroactive (Definition~\ref{def:retroactivity}).
\item \emph{Privacy.} $\A$ is $(2W\eps_0,\,2We^{\eps_0}\delta_0)$-private (Definition~\ref{def:privacy}); moreover, for every $\delta'>0$ it is $(\eps_W,\delta_W)$-private with
\[
\eps_W=2\eps_0\sqrt{2W\ln(1/\delta')}+2W\eps_0(e^{2\eps_0}-1)\,,\qquad \delta_W=2We^{\eps_0}\delta_0+\delta'.
\]
\item \emph{Amplification by the secrecy of the sample.} Let $n\ge k$ be an occupancy parameter\footnote{$n$ is a parameter of the guarantee, not of the algorithm.}, denote $q=k/(n+1)$, and assume $\eps_0\le1$. Restricted to neighboring pairs in which both sequences satisfy $n_\ell\ge n$ for all $\ell$, $\A$ is, for every $\delta'>0$, $(\eps_{W,q},\delta_{W,q})$-private with
\begin{align*}
\eps_{W,q}&=\eps_q\sqrt{2W\ln(1/\delta')}+W\eps_q(e^{\eps_q}-1)\,,\qquad \delta_{W,q}=2Wq\,e^{\eps_0}\delta_0+\delta'\,,\\
\eps_q&=\ln\big(1+q(e^{2\eps_0}-1)\big)\le 7q\eps_0 .
\end{align*}
\end{enumerate}
\end{theorem}

\begin{proof}
\emph{Retroactivity.} The sample $Q_\ell$ is a fresh coin function of $\Live_\ell(S)$, which is determined by $\Data_\ell(S)$, and $z_\ell$ is a fresh coin function of $Q_\ell$. Hence the answer block from any round $i$ onward is per round post processing, with fresh coins, of $\{\Data_\ell(S)\}_{\ell=i}^T$, and Lemma~\ref{lem:closure} gives retroactivity.

\emph{Privacy.} As in the privacy proof of Theorem~\ref{thm:linear}, one of the two neighboring sequences, denote it $S^+$ and the other $S^-$, holds one extra surviving copy of a stamped point $(t^*,x^*)$ on a contiguous set of rounds, the data agreeing otherwise; intersected with liveness, this leaves at most $W$ rounds at which $\Live_\ell(S^+)$ exceeds $\Live_\ell(S^-)$ by one copy of $x^*$, the two windows agreeing elsewhere. At every other round the two samples have the same distribution. At each of the at most $W$ differing rounds, a uniformly random subset of $\Live_\ell(S^+)$ of the required size can be drawn by first drawing a uniformly random subset of $\Live_\ell(S^-)$ of the same size and then, with probability $q_\ell$ equal to the sample size divided by $|\Live_\ell(S^+)|$, replacing a uniformly random element of it by the extra copy (or, when the sample is the whole window, adding the extra copy). So the two samples differ by at most one removal and one insertion, and by group privacy of $M$ at distance two the round distributions are $(2\eps_0,2e^{\eps_0}\delta_0)$-indistinguishable. Since all coins are fresh, the two transcripts are product distributions differing in at most $W$ factors, and basic composition, respectively advanced composition \cite{DRV10}, gives the two claimed bounds.

\emph{Amplification.} Under the assumption that $n_\ell\ge n$ for all $\ell$, every sample has size $k$, so at each differing round the coupling above replaces an element with probability $q_\ell=k/(n_\ell(S^-)+1)\le q$, and with the remaining probability the two samples coincide. The round distributions are therefore $(1-q_\ell)P+q_\ell P'$ versus $P$, with $P'$ and $P$ $(2\eps_0,2e^{\eps_0}\delta_0)$-indistinguishable, and such a mixture is $\big(\ln(1+q_\ell(e^{2\eps_0}-1)),\,2q_\ell e^{\eps_0}\delta_0\big)$-indistinguishable from $P$, in both directions, by the standard argument for amplification by subsampling \cite{BBG18}. Composing over the at most $W$ differing rounds as before gives the claim; the bound $\eps_q\le7q\eps_0$ for $\eps_0\le1$ follows from $e^{2\eps_0}-1\le(e^2-1)\eps_0$ by convexity.
\end{proof}

The amplified privacy of part 3 holds only under the occupancy promise that $n_\ell\ge n$ for all $\ell$. The next lemma removes it: the algorithm privately tests, using the machinery of Section~\ref{sec:linear}, whether the window is large enough, and releases $\bot$ when it is not.

\begin{lemma}[Removing the occupancy promise]\label{lem:gate}
Assume, as in part 3 of Theorem~\ref{thm:sampling}, that $\eps_0\le1$ and $n\ge k$, let $q=k/(n+1)$, and let $\eps_{W,q}$ and $\delta_{W,q}$ be as defined there. Let $\eps_t,\delta_t>0$, and let $\tilde n_\ell$ be the release of Theorem~\ref{thm:linear} applied to the window count $n_\ell$ with Laplace noise, calibrated to be $(\eps_t,0)$-private.  Define the {\em gated algorithm} that releases $\bot$ at rounds with $\tilde n_\ell<2n$ and $M(Q_\ell)$ at all other rounds.\footnote{So here the algorithm itself depends on the occupancy parameter $n$, unlike in part 3 of Theorem~\ref{thm:sampling}.} Then the gated algorithm is exactly retroactive. Moreover, there is a universal constant $C$ such that if $n\ge C\log^2(T)\log(T/\delta_t)/\eps_t$, then it is $(\eps_t+\eps_{W,q},\,\delta_{W,q}+\delta_t)$-private for \emph{every} pair of neighboring sequences, and at every round with $n_\ell\ge3n$ it releases $M(Q_\ell)$, except with probability $\delta_t$.
\end{lemma}

\begin{proof}
The window count is a linear window statistic, so the count release is exactly retroactive by Theorem~\ref{thm:linear} applied to the window histogram, and the gated releases are per round post processing, with fresh coins, of the count release and of $\Live_\ell(S)$; Lemma~\ref{lem:closure} gives retroactivity. The count error is at most $n-1$ at all rounds, except with probability $\delta_t$: the release at each round is a sum of at most $\log_2(T)+1$ Laplace variables of scale $O(\log(T)/\eps_t)$ each, and a union bound over the rounds and the summands bounds all errors by $O(\log^2(T)\log(T/\delta_t)/\eps_t)\le n-1$, by the choice of $C$. 

For privacy, passing the gate at a round $\ell$ requires $\tilde n_\ell\ge2n$, hence $n_\ell\ge2n-(n-1)=n+1$ unless the count fails, which happens with probability at most $\delta_t$ under either sequence. The count release is $(\eps_t,0)$-private. Conditioned on its outcome, the set of passing rounds is fixed. The releases at these rounds are then exactly those of part 3 of Theorem~\ref{thm:sampling}, restricted to the passing rounds, and the per round amplified bound applies at each of them.
Composing the count release with the gated releases, and adding the failure probability of the count, gives the claim. Finally, at a round with $n_\ell\ge3n$ we have $\tilde n_\ell\ge3n-(n-1)\ge2n$ unless the count fails.
\end{proof}

We conclude with a concrete instantiation: privately reporting a point that lies between the minimum and the maximum of the live points, the \emph{interior point} problem, which is the basic primitive underlying private medians, quantiles, and learning of thresholds. Combining the compiler with the gate and with the one shot algorithm of \cite{CLNSS23} yields the following.

\begin{corollary}[A window interior point]\label{cor:window-ipp}
Instantiate $M$ with the one shot interior point algorithm of \cite[Theorem~1.1]{CLNSS23}, run with $\eps_0=1$, and set $k$ to its sample complexity, $k=\widetilde O(\log^*|\X|)\cdot\mathrm{polylog}(1/\delta_0,1/\beta)$. Let $\eps_t,\delta_t,\delta'>0$, and let $n$ be as in Lemma~\ref{lem:gate} (growing as $\mathrm{polylog}(T)/\eps_t$). Then the gated algorithm of Lemma~\ref{lem:gate} is exactly retroactive, is $\big(\eps_t+O(\tfrac{k}{n}\sqrt{W\ln(1/\delta')}+W(\tfrac{k}{n})^2),\,O(\tfrac{k}{n}W\delta_0)+\delta'+\delta_t\big)$-private\footnote{These are the parameters $(\eps_t+\eps_{W,q},\,\delta_{W,q}+\delta_t)$ of Lemma~\ref{lem:gate} with $\eps_0=1$ and $q\le k/n$, after plugging in $\eps_{W,q}$ and $\delta_{W,q}$.} for every pair of neighboring sequences, and at every round with $n_\ell\ge3n$ releases, with probability $1-\beta-\delta_t$, a value between the minimum and the maximum of the live points; at all other rounds it releases either such a value or $\bot$.
\end{corollary}

\begin{proof}
Retroactivity and privacy follow from Lemma~\ref{lem:gate} with $\eps_0=1$ and $q\le k/n$. Whenever the gate passes, the sample has size $k$ (unless the count fails, passing implies $n_\ell\ge n+1>k$), and an interior point of $Q_\ell\subseteq\Live_\ell(S)$ is an interior point of $\Live_\ell(S)$, so $M$ succeeds with probability $1-\beta$; at rounds with $n_\ell\ge3n$ the gate passes except with probability $\delta_t$.
\end{proof}

\begin{remark}[Persistent ranks]\label{rem:persistent}
A natural variant attaches to each addition, upon arrival, a persistent secret rank drawn uniformly from $[0,1]$, and takes as $Q_\ell$ the $\min\{k,n_\ell\}$ live copies of smallest rank. The sample then changes only with the window's actual turnover, by at most one swap per arrival, expiry, or deletion. This can be helpful when the releases are used to estimate change over time, and when $M$ is expensive and is maintained incrementally on the sample. This variant is also exactly retroactive, and it satisfies part 2 (privacy) of Theorem~\ref{thm:sampling}. Its amplification guarantees, however, are weaker and somewhat more complicated. Appendix~\ref{app:amplification} gives the precise statements and proofs.
\end{remark}

\section{On the price of retroactivity}\label{sec:separation}

In this section we prove a negative result on the cost of retroactivity when enforced on top of privacy. Specifically, we show that retroactivity forbids instance adaptive error bounds for the {\em counting distinct elements} problem, of the type studied by Jain et al.~\cite{JKRSS23}.
Admittedly, our negative result provides only an instance adaptive separation, rather than a worst case separation, which we leave as an open question.

\paragraph{The CountDistinct task.}
Let $\X$ be a large universe and suppose there is at most one update per round, as in Subsection~\ref{subsec:clustering}. The functionality of interest is the number of distinct surviving values, i.e., for a multiset of stamped points $D$,
\[
\mathrm{CD}(D)=|\{x\in\X: \exists t \text{ such that } (t,x)\in D\}|.
\]

\begin{definition}[\cite{JKRSS23}]
For a sequence $S$ and a value $x\in\X$, let $p_x(\ell)\in\{0,1\}$ indicate whether $x$ is present in $\Data_\ell(S)$, with $p_x(0)=0$. The \emph{flippancy} of $x$ is the number of rounds at which the indicator $p_x$ changes, and the \emph{maximum flippancy} $w(S)$ is the largest flippancy of any value.
\end{definition}

Note that $w(S)$ is determined by the trajectory $\{\Data_\ell(S)\}_{\ell=1}^T$, and that every insertion only sequence has maximum flippancy $1$. We say a dynamic algorithm is \emph{$(\alpha,\beta)$-accurate per round} on a class of sequences if for every sequence $S$ in the class and every round $\ell$, $\Pr[|a_\ell-\mathrm{CD}(\Data_\ell(S))|>\alpha]\le\beta$.

Jain et al.~\cite{JKRSS23} presented a continually-DP\footnote{The algorithm of Jain et al.~\cite{JKRSS23} guarantees a stronger notion of privacy, called item-level (continual) DP, allowing all events of a single item to change at once. Their guarantee implies the event level notion of privacy studied in this paper.} mechanism for this problem that adapts to the data it actually sees by tracking how frequently items change (the flippancy, $w(S)$). It increases its noise scale only when these changes double, guaranteeing an error of
\begin{equation}\label{eq:Jain}
\widetilde{O}\Big(\big(\sqrt{w(S)}\log T+\log^3T\big)\cdot\sqrt{\log(1/\delta)}\,/\,\eps\Big)    
\end{equation}
on \emph{every} stream, with no prior bound on $w(S)$ (Theorem 1.5 in \cite{JKRSS23}). Because the error scales with actual changes, streams with very few changes (maximum flippancy $O(1)$) have a much smaller, polylogarithmic error.

However, their mechanism is not retroactive. The amount of noise it adds is governed by the flippancy of the entire history, including flips of copies long since deleted, which is not a function of the surviving trajectory. The theorem below shows this is not an accident of their design, but a strict limitation of retroactivity itself.

We first record what retroactivity remembers, or rather what it is required to forget.

\begin{lemma}[Amnesia]\label{lem:amnesia}
Let $\A$ be a retroactive dynamic algorithm. Then for every round $\ell$, the distribution of $a_\ell$ on input $S$ depends only on $\ell$ and on the stamped multiset $\Data_\ell(S)$. Consequently, if $\A$ is $(\alpha,\beta)$-accurate per round on insertion only sequences, then it is $(\alpha,\beta)$-accurate per round on all sequences.
\end{lemma}

\begin{proof}
To see that the output distribution at round $\ell$ depends only on the current dataset $\Data_\ell(S)$, consider two sequences $S$ and $S'$ that result in the exact same state at round $\ell$, i.e., $\Data_\ell(S)=\Data_\ell(S')$. We can bridge them using a hybrid sequence $S_{\mathrm{mid}}$ that takes its first $\ell$ updates from $S'$ and its remaining updates from $S$. Now,
\begin{itemize}
    \item Since the input sequences $S'$ and $S_{\mathrm{mid}}$ are identical up to time $\ell$, by causality, we have that $a_\ell(S')\equiv a_\ell(S_{\mathrm{mid}})$.
    \item Since the states of $S$ and $S_{\mathrm{mid}}$ agree from round $\ell$ onward, by retroactivity we have that $a_\ell(S)\equiv a_\ell(S_{\mathrm{mid}})$.
\end{itemize}
This shows that the output depends only on the current data state. For the consequence, the state of any sequence at round $\ell$ is replicated by an insertion only sequence, placing at every round $t\le\ell$ exactly the additions of the copies of $\Data_\ell(S)$ carrying stamp $t$. The output distributions at time $\ell$ are thus equal, and so accuracy on the insertion-only sequence implies accuracy on the insertion-deletion sequence.
\end{proof}

We now leverage Lemma~\ref{lem:amnesia} to show our impossibility result. At a high level, our impossibility result is obtained as follows. Suppose towards contradiction that there is a private and retroactive algorithm $\A$ for $\mathrm{CountDistinct}$, with an input adaptive error bound depending on the flippancy $w(S)$, scaling as $\mathrm{poly}(w(S),\log T)$, in the spirit of Equation~(\ref{eq:Jain}). So, on an insertion-only sequence $S$ (with flippancy 1) it guarantees a small error of ${\rm polylog}(T)$. Then, by Lemma~\ref{lem:amnesia}, it also guarantees this small error on {\em any} input sequence. This violates an impossibility result of Jain et al.~\cite{JKRSS23}, showing that any (event level) DP algorithm for $\mathrm{CountDistinct}$, even non-retroactive, must have worst case error $\Omega(T^{1/4})$.

A technical issue with formalizing this proof sketch is that, as stated, the results of Jain et al.~\cite{JKRSS23} only rule out DP algorithms that guarantee {\em simultaneous} accuracy across all time steps together. That is, algorithms guaranteeing that with constant probability {\em all} of their answers are accurate simultaneously. In contrast, Lemma~\ref{lem:amnesia} only guarantees {\em per-round marginal} accuracy. Fortunately, straightforward modifications to the negative result of Jain et al.~\cite{JKRSS23} lift it also to marginally-accurate DP algorithms, as captured in the following theorem (the proof is given in Appendix~\ref{app:reconstruction} for completeness). 

\begin{theorem}[Marginal worst case bound \cite{JKRSS23}]\label{thm:marginal-lb}
There are constants $c,\beta^*,\delta^*>0$ such that for all sufficiently large $T$, no dynamic algorithm for $\mathrm{CountDistinct}$ that is $(1,\delta^*)$-private (Definition~\ref{def:privacy}) is $(c\,T^{1/4},\beta^*)$-accurate per round on the class of all input sequences.
\end{theorem}

We are now ready to state and prove our impossibility result.

\begin{theorem}[Retroactivity forbids instance adaptive error]\label{thm:countdistinct}
There are constants $c,\beta^*,\delta^*>0$ such that the following holds for all sufficiently large $T$. Let $\A$ be a dynamic algorithm for $\mathrm{CountDistinct}$ that is $(1,\delta^*)$-private (Definition~\ref{def:privacy}) and retroactive (Definition~\ref{def:retroactivity}), and suppose that its error is bounded in terms of the maximum flippancy: for some function $f$, for every input sequence $S$ and every round $\ell$,
\[
\Pr\big[|a_\ell-\mathrm{CD}(\Data_\ell(S))|>f(w(S))\big]\;\le\;\beta^*.
\]
Then $f(1)\ge c\,T^{1/4}$. 
\end{theorem}

In words, Theorem~\ref{thm:countdistinct} states that the instance-adaptive bound is already at worst-case level on the easiest instances. In contrast, the mechanism of \cite{JKRSS23}, which is private but not retroactive, satisfies the above with $f(w)=\widetilde O\big((\sqrt w\log T+\log^3T)\sqrt{\log(1/\delta)}/\eps\big)$, which is polylogarithmic at $w=1$. 

\begin{proof}[Proof of Theorem~\ref{thm:countdistinct}]
Every insertion only sequence has maximum flippancy $1$, so $\A$ is $(f(1),\beta^*)$-accurate per round on insertion only sequences. By Lemma~\ref{lem:amnesia}, the answer distribution at any round of any sequence coincides with the answer distribution at the same round of an insertion only sequence realizing the same state, so $\A$ is $(f(1),\beta^*)$-accurate per round on all sequences. Theorem~\ref{thm:marginal-lb} then gives $f(1)>c\,T^{1/4}$.
\end{proof}

\begin{remark}[Forced budget amnesia]\label{rem:amnesia}
The result gives precise content to an intuition one might call \emph{budget amnesia}. Instance-adaptive continual mechanisms, such as the flippancy tracker of \cite{JKRSS23}, keep books: a data dependent account of the input's past activity, which lets them spend noise only on the activity that actually occurred. Lemma~\ref{lem:amnesia} says a retroactive algorithm cannot keep such books: its answer distribution at each round is required to forget every flip that did not survive into the current data. The bookkeeping is not merely reset by deletions; it is forbidden by the definition, and Theorem~\ref{thm:countdistinct} is the price.
\end{remark}

{\small
\bibliographystyle{alpha}

\newcommand{\etalchar}[1]{$^{#1}$}

}

\appendix

\section{A model with retroactive insertions}\label{app:insertions}

We describe the extension of the model of Section~\ref{sec:model} in which additions, too, may reach into the past, recovering the full update generality of \cite{DIL07}. An update is now either an addition $(\add,t,x)$, which retroactively inserts the point $x\in\X$ at time $t$, or a deletion $(\del,t,x)$ as before. The model of Section~\ref{sec:model} is the special case in which every addition arriving at time step $i$ carries the stamp $t=i$.

The current data $\Data_i(S)$ is again the multiset defined by processing $u_1,\dots,u_i$ in \emph{arrival} order, starting from the empty multiset. For every $u_j$: first, for every $(\add,t,x)\in u_j$ with $t\le j$, add a copy of $(t,x)$ to $\Data_i(S)$ (additions with $t>j$, naming a future time, are ignored); then, for every $(\del,t,x)\in u_j$, if $(t,x)\in\Data_i(S)$ then delete one copy of $(t,x)$ from $\Data_i(S)$ (and otherwise do nothing). As in Remark~\ref{rem:model}, the effect of an update is determined at its arrival and updates are never revisited; in particular, a deletion arriving before the addition it names is a permanent no-op, even if a matching addition arrives later.

The neighboring relation extends accordingly, still at the granularity of a single operation.

\begin{definition}[Neighboring sequences, with retroactive insertions]\label{def:neighboring-app}
Input sequences $S,S'$ are \emph{neighboring} if one of them can be obtained from the other by inserting a single update into some $u_j$: either an addition $(\add,t,x)$ with $t\le j$, or a deletion $(\del,t,x)$.
\end{definition}

Definitions~\ref{def:privacy}, \ref{def:uniting}, \ref{def:retroactivity}, and~\ref{def:main}
extend verbatim to the richer update language. At time $i$, after applying all deletions
and retroactive insertions seen so far, the revised state of a past time $\ell\le i$ is
\[
\mathrm{RevData}_{i,\ell}(S)
=
\{(t,x)\in\Data_i(S):t\le\ell\}.
\]
Thus, partial retroactivity queries only the current state $\Data_i(S)$, whereas full
retroactivity may also query any revised past state $\mathrm{RevData}_{i,\ell}(S)$; both are therefore analyzed under the same privacy definition. The main additional subtlety concerns uniting. A retroactive insertion does not affect any state before the round in which it arrives, but once it arrives, it changes every revised past state whose time is at least its time-stamp. Hence, it is not enough for two histories to contain the same current points: they must also agree on the revised
time-stamps of all surviving retroactive insertions. For example, if the same surviving point has stamp $6$ in one history and stamp $8$ in the other, then their revised states differ for every $6\le\ell<8$, even though their current sets may be identical.

\section{Adaptive privacy}\label{sec:adaptive}

Definition~\ref{def:privacy} treats the input stream as fixed in advance: the stream is oblivious to the answers the algorithm releases. In this appendix we strengthen the definition to streams chosen adaptively as a function of past answers, in the spirit of the adaptive continual release model of Jain et al.~\cite{JRSS23}, adapted to our update terminology and to our single operation neighboring relation. The adaptive adversary supplies the regular updates in every round and, once during the game, designates a challenge update: a single addition or a single deletion. The challenge is present in one world and absent in the other, and privacy requires that the adversary's view be insensitive to which world it is playing in.

Formally, for a dynamic algorithm $\A$, an adversary $\B$, and a bit $b\in\{0,1\}$, consider the following game, denoted $\mathsf{AdaptGame}_b(\A,\B)$.
\begin{enumerate}[itemsep=1pt,topsep=3pt]
\item For $i=1,2,\dots,T$:
\begin{enumerate}[itemsep=1pt,topsep=1pt]
\item Based on $(a_1,\dots,a_{i-1})$ and its internal randomness, the adversary $\B$ outputs a multiset $u_i$ of updates. In addition, $\B$ may output one of the following two declarations, each of which may be made at most once throughout the entire game:
\begin{enumerate}[itemsep=1pt,topsep=1pt]
\item a \emph{challenge addition} $(\add,x^*)$ for a point $x^*\in\X$ of its choice; or
\item a \emph{challenge deletion} $(\del,t^*,x^*)$ naming a round $t^*<i$ and a point of its choice.
\end{enumerate}
\item If $b=0$, let $\tilde{u}_i$ be $u_i$ together with the challenge update declared in this round (if any); if $b=1$, let $\tilde{u}_i=u_i$.\footnote{In world $b=1$ the challenge update is simply omitted; the challenge deletion may name a regular addition, or nothing at all.}
\item Algorithm $\A$ receives $\tilde{u}_i$ and returns an answer $a_i$, which is given to $\B$.
\end{enumerate}
\item The view of the adversary is $\mathsf{View}_b(\A,\B)=\big(\text{the internal randomness of }\B,\ a_1,\dots,a_T\big)$.
\end{enumerate}

\begin{definition}[$(\eps,\delta)$-adaptive privacy]\label{def:adaptive-privacy}
A dynamic algorithm $\A$ is \emph{$(\eps,\delta)$-adaptively private} if for every adversary $\B$ and every event $E$,
\[
\Pr\big[\mathsf{View}_0(\A,\B)\in E\big]\;\le\;e^{\eps}\cdot\Pr\big[\mathsf{View}_1(\A,\B)\in E\big]+\delta,
\]
and symmetrically with the roles of the two worlds exchanged.
\end{definition}

\begin{definition}[Adaptively private retroactive algorithm]\label{def:main-adaptive}
A dynamic algorithm $\A$ is an \emph{$(\eps,\delta)$-adaptively-private retroactive algorithm} if it is $(\eps,\delta)$-adaptively private (Definition~\ref{def:adaptive-privacy}) and retroactive (Definition~\ref{def:retroactivity}).
\end{definition}

\section{The persistent ranks variant}\label{app:amplification}

In the standard sampling compiler of Theorem~\ref{thm:sampling}, the algorithm draws a fresh, independent sample at every round. In this appendix, we explain the persistent-ranks variant of Remark~\ref{rem:persistent}, using the notation of Subsection~\ref{subsec:sampling}. In this variant, each addition receives a single secret rank upon arrival, which dictates its inclusion in the sample for its entire lifespan. This creates strict temporal correlation: an extra item's inclusion is no longer an independent coin flip at each round, meaning standard per-round amplification breaks down. To resolve this, the analysis shifts from a per-round perspective to a window-level analysis. The goal is to show that, during the window in which the two executions differ on this item, the entire transcript is affected only if its persistent rank falls in a low-probability range. For deletions, a cyclic rank-shifting argument lets us reduce the analysis to a single additional random rank, rather than having to track a changing collection of affected items over time.

The algorithm $\A^{\mathrm{rank}}_{M,k,W}$ attaches to each addition, upon arrival, an independent rank drawn uniformly from $[0,1]$, kept secret and persistent; when a deletion removes one of several copies of a stamped point, the copy instantiated last is removed. At every round $\ell$ it forms $Q_\ell$, the multiset of values of the $\min\{k,n_\ell\}$ live copies of smallest rank, and releases $M(Q_\ell)$ with fresh independent coins. Let $M$ be $(\eps_0,\delta_0)$-differentially private, and let $(\eps_W,\delta_W)$ be as in Theorem~\ref{thm:sampling}.

\paragraph{Retroactivity.}
Fix $S,S'$ uniting at time $i$. The block of samples $(Q_\ell)_{\ell\ge i}$ is a randomized function of $\{\Data_\ell(S)\}_{\ell=i}^T$, and is otherwise independent of $S$: the sample at round $\ell$ is determined by $\Live_\ell(S)$, itself determined by $\Data_\ell(S)$, and by the ranks of its elements, which are i.i.d.\ uniform variables (regardless of the rounds at which they were instantiated).\footnote{Nitpicking: when deleting one of multiple copies of the same stamped point, in which case the algorithm has freedom in choosing which point to delete (see Remark~\ref{rem:model}), this choice must not depend on the ranks for them to remain independent. Deleting points by instantiation order, as our algorithm does, works.} So, as the two trajectories coincide from $i$ on, the sample blocks during the two executions are identically distributed, and thus so are the answer blocks, which are per round post processing of them (see Lemma~\ref{lem:closure}).

\paragraph{Privacy.}
Let $S,S'$ be neighboring. As in the privacy proof of Theorem~\ref{thm:linear}, one of the two neighboring sequences holds one extra surviving copy of a fixed stamped point $(t^*,x^*)$ on a contiguous set of rounds, the data agreeing otherwise. This extra point may influence the outcome only during a window of at most $W$ rounds, after which it expires. Now, for every fixing of the ranks and for every round in this window, the samples $Q,Q'$ during the two executions may differ by at most one addition and one removal of stamped points. (This is because the extra stamped point, with its rank, may enter the sample and push another out of it.) 
By group privacy of $M$ at distance two, throughout the window, the per round distributions are $(2\eps_0,\,2e^{\eps_0}\delta_0)$-indistinguishable. Basic composition across the window gives $(2W\eps_0,2We^{\eps_0}\delta_0)$-indistinguishability, while advanced composition \cite{DRV10} gives $(\eps_W,\delta_W)$. As this holds for every fixing of the ranks, this also holds without the fixing.

\paragraph{Amplification.}
The secrecy of the ranks improves these parameters on windows that are much larger than the sample. The following theorem quantifies this improvement.

\begin{theorem}\label{thm:amplification}
Let $M$ be $(\eps_0,\delta_0)$-differentially private, let $n\ge k$ and $\beta_r\in(0,1)$, and set $\bar p=\min\big(1,\,(2k+8\ln(W/\beta_r))/n\big)$. Assume $\eps_W\le1$ and $\bar p\le1/2$. Then for every pair of neighboring sequences in which both satisfy $n_\ell\ge n$ for all $\ell$, the transcripts of $\A^{\mathrm{rank}}_{M,k,W}$ are $(12\bar p\,\eps_W,\;3\bar p\,\delta_W+\beta_r)$-indistinguishable.
\end{theorem}

\begin{proof}
\emph{Setup.} As in the privacy proof of Theorem~\ref{thm:sampling}, one of the two sequences, denote it $S^+$, holds one extra surviving copy of a stamped point $(t^*,x^*)$ at a contiguous set of rounds, the data agreeing otherwise. Intersecting this one-copy difference with the live window gives a set $W^*$ of at most $W$ rounds at which $\Live_\ell(S^+)$ exceeds $\Live_\ell(S^-)$ by one copy of $x^*$, the two live windows agreeing elsewhere. To obtain amplification over the whole window, we would like the discrepancy between the two executions to be governed by a single random rank throughout $W^*$. The difficulty is that, after common deletions, the copy that is unmatched between the two executions may change. If each such copy carried its own independent rank, different parts of the window could be governed by different random ranks. We therefore couple the ranks so that throughout $W^*$ the live-rank multisets differ by one fixed rank $r^*$, with all remaining ranks matched between the two executions. 

We couple the ranks of the two executions so that throughout $W^*$ the one-rank difference is a fixed rank $r^*$, with all remaining ranks matched between the two executions. For a differing addition this is immediate: we place the added copy first among copies of the same stamped point, so its rank $r^*$ remains the one-rank difference under subsequent common LIFO deletions. For example, if before a common deletion the ranks are
\[
\begin{array}{c@{\qquad}c}
S^- & S^+\\[2mm]
(r_1,r_2,\ldots,r_m)
&
(r^*,r_1,r_2,\ldots,r_m),
\end{array}
\]
then the common LIFO deletion removes $r_m$ from both executions, leaving
\[
\begin{array}{c@{\qquad}c}
S^- & S^+\\[2mm]
(r_1,r_2,\ldots,r_{m-1})
&
(r^*,r_1,r_2,\ldots,r_{m-1}).
\end{array}
\]

Since the ranks are persistent and may already affect the releases before the differing deletion, the coupling must be defined consistently already before that deletion; it is not enough to couple only the states that remain afterwards. For a differing addition, this is easy: the coupling can place the unmatched rank at one end, while subsequent common deletions proceed from the other end. For a differing deletion, however, the differing deletion and the subsequent common deletions act according to the same deletion rule, so under such a simple coupling the unmatched rank may change after a common deletion. 

Thus, for a differing deletion, $S^+$ is the sequence without the extra deletion. Let $m$ be the number of surviving copies of $(t^*,x^*)$ immediately before the extra deletion, indexed by instantiation order. Under $S^-$, we label their ranks $r_1,\dots,r_m$. Under $S^+$, we couple the ranks cyclically, assigning $r_m$ to the first copy and, for every $i\ge2$, $r_{i-1}$ to the $i$th. In both executions we designate $r_m$ as $r^*$. Thus, immediately before the differing deletion, the ranks of these copies are coupled as

$$
\begin{array}{c@{\qquad}c}
S^- & S^+\\[2mm]
(r_1,r_2,\ldots,r_{m-1},r_m)
&
(r_m,r_1,\ldots,r_{m-2},r_{m-1}).
\end{array}
$$

The differing LIFO deletion acts only in $S^-$ and removes its last copy, which carries $r_m=r^*$. Hence, immediately afterwards,

$$
\begin{array}{c@{\qquad}c}
S^- & S^+\\[2mm]
(r_1,r_2,\ldots,r_{m-1})
&
(r^*,r_1,r_2,\ldots,r_{m-1}).
\end{array}
$$

Every subsequent common LIFO deletion, while the one-copy difference persists, removes the same rightmost rank $r_j$ from both executions, so the difference between their live rank multisets remains $r^*$.\footnote{This is a valid coupling because a fixed relabeling of i.i.d.\ ranks preserves each marginal distribution.}
\medskip

\emph{Inclusion patterns.}
The coupling above resolves the difficulty of keeping a single distinguished rank throughout the differing window. A separate issue remains. While in both the addition and deletion cases $r^*$ is uniform and independent of the ranks $\rho$ of all other copies, in the deletion case $r^*$ is present in both executions before the differing deletion, and may therefore already affect the released answers before the two executions diverge. To account for this dependence, we track the inclusion of $r^*$ over its entire lifespan, rather than only during the differing window.

Let $L^*$ be the full \emph{lifespan} of this extra copy under $S^+$ (the set of rounds where it is live). Since items expire after $W$ rounds, $|L^*| \le W$. Note that $L^*$ fully includes $W^*$ (the critical window of rounds where the extra copy is live under $S^+$ but absent under $S^-$). 

To determine if this extra copy enters the sample at a given round $\ell \in L^*$, we compare its rank to those of the other available items. Let $\tau_\ell$ be the $k$-th smallest rank among all live copies \emph{excluding} the extra copy. (Notice that at rounds $\ell \in W^*$, these other copies are exactly the items comprising $\Live_\ell(S^-)$)\footnote{If fewer than $k$ other live copies are present, we set
$\tau_\ell=+\infty$, so that
$I_\ell(r^*)=\mathbf{1}[r^*<\tau_\ell]$ continues to describe the
inclusion decision exactly. This convention can arise only outside
$W^*$, since for $\ell\in W^*$ there are at least $n\ge k$ other live
copies.}. 

Given the fixed ranks of all other items (denoted $\rho$), the threshold $\tau_\ell$ is a fixed value. The extra copy enters the sample at round $\ell$ if and only if its rank beats the threshold: $I_\ell(r^*) = \mathbf{1}[r^* < \tau_\ell]$. Consequently, the entire transcript of answers under $S^+$ depends on $r^*$ exclusively through the binary vector $J(r^*) = (I_\ell(r^*))_{\ell \in L^*}$, which we call the \emph{inclusion pattern}. Let $\Pi(\cdot \mid J)$ denote the conditional distribution of the transcript given $\rho$ and a specific pattern $J$. 

Under the neighboring sequence $S^-$, the extra copy is completely absent during $W^*$. Therefore, its effective inclusion pattern, $J^-(r^*)$, is forced to zero across $W^*$, but matches $J(r^*)$ everywhere else in $L^*$. The transcript under $S^-$ thus follows the conditional distribution $\Pi(\cdot \mid J^-(r^*))$.

As a general property, if \emph{any} two arbitrary patterns differ at exactly $h$ rounds, their corresponding samples differ by one item substitution (one removal and one insertion) at each of those $h$ rounds. Because the base mechanism $M$ uses fresh coins each round, we can compose the privacy loss exactly as in Theorem~\ref{thm:sampling}, meaning these sample differences result in conditional distributions that are $(2h\eps_0,\,2he^{\eps_0}\delta_0)$-indistinguishable. Since the lifespan $|L^*| \le W$, \emph{any} two arbitrary patterns over $L^*$ have $(\eps_W,\delta_W)$-indistinguishable conditional distributions via advanced composition.

\medskip

\emph{Amplification.} We now bound how much the probability of any arbitrary output event $E$ changes between the two universes. Fix $\rho$, and let $\tau^* = \max_{\ell\in W^*}\tau_\ell$ be the threshold the extra item must beat to enter the sample at least once during the critical window. Since $r^*$ is drawn uniformly from $[0,1]$, let $p = \tau^*$ be the probability that the random rank $r^*$ successfully beats this threshold (i.e., $r^* < \tau^*$).

\textbf{Case 1 ($r^* \ge \tau^*$):} The extra item's rank is too high, meaning it never enters the sample during $W^*$ under $S^+$. Because its pattern was already forced to zero during $W^*$ under $S^-$, the two inclusion patterns are perfectly identical ($J(r^*) = J^-(r^*)$), and the conditional transcripts coincide perfectly. Thus, $\Pi(E\mid J(r^*)) - \Pi(E\mid J^-(r^*)) = 0$.

\textbf{Case 2 ($r^* < \tau^*$):} The extra item enters the sample at least once during $W^*$ under $S^+$. In this case, applying our general pattern bound gives $\Pi(E\mid J(r^*))\le e^{\eps_W}\Pi(E\mid J^-(r^*))+\delta_W$. 

To relate this back to the overall probability of $E$ under $S^-$, we must introduce a hypothetical "bad" rank $\tilde r \ge \tau^*$. Under $S^-$, the inclusion patterns $J^-(r^*)$ and $J^-(\tilde r)$ might differ (for example, at rounds outside of $W^*$, since the two rank values may induce different inclusion decisions prior to the differing deletion). Fortunately, because \emph{any} two patterns over $L^*$ are $(\eps_W, \delta_W)$-indistinguishable, we can safely bound $\Pi(E\mid J^-(r^*))\le e^{\eps_W}\Pi(E\mid J^-(\tilde r))+\delta_W$.

We choose $\tilde r \ge \tau^*$ to minimize this latter probability. 
Because the minimum over $\tilde r\ge\tau^*$ is always less than or equal to the conditional average over the same range, it is bounded by the conditional average of $\Pi(E\mid J^-(\tilde r))$ over all $\tilde r \ge \tau^*$. This conditional average is at most $\frac{\Pr_{S^-}[E\mid\rho]}{1-p}$. Substituting this bound back into our bound gives:
\begin{equation}\label{eq:avg-bound}
\Pi(E\mid J^-(r^*)) \le e^{\eps_W}\left(\frac{\Pr_{S^-}[E\mid\rho]}{1-p}\right)+\delta_W.
\end{equation}

Next, we rearrange our initial Case 2 bound to isolate the difference between the two universes:
\begin{equation}\label{eq:diff-bound}
\Pi(E\mid J(r^*)) - \Pi(E\mid J^-(r^*)) \le (e^{\eps_W}-1)\Pi(E\mid J^-(r^*)) + \delta_W
\end{equation}

By plugging \eqref{eq:avg-bound} into the right side of \eqref{eq:diff-bound}, we obtain a worst-case bound for the difference when $r^* < \tau^*$. Because this difference is exactly zero in Case 1, taking the expectation over all possible values of $r^*$ simply multiplies our worst-case difference by $p$ (the probability of Case 2). Therefore:
\begin{align*}
\Pr_{S^+}[E\mid\rho]-\Pr_{S^-}[E\mid\rho]
&\;=\;\mathbb{E}_{r^*}\big[\mathbf 1_{\{r^* < \tau^*\}}\big(\Pi(E\mid J(r^*))-\Pi(E\mid J^-(r^*))\big)\big]\\
&\;\le\;p\Big[(e^{\eps_W}-1)\Big(\frac{e^{\eps_W}\Pr_{S^-}[E\mid\rho]}{1-p}+\delta_W\Big)+\delta_W\Big].
\end{align*}

We next show that $p\le \bar p$ with high probability. Since $p=\max_{\ell\in W^*}\tau_\ell$, it suffices to bound $\tau_\ell$ for each $\ell\in W^*$. Fix such a round $\ell$. Since $\tau_\ell$ is the $k$th smallest rank among the live copies other than the extra copy, the event $\tau_\ell>\bar p$ occurs exactly when fewer than $k$ of these ranks fall below $\bar p$ (less than $k$ successes). At rounds $\ell\in W^*$ these copies are precisely those in $\Live_\ell(S^-)$, and by assumption $|\Live_\ell(S^-)|\ge n$. Therefore, $$\Pr_\rho[\tau_\ell>\bar p]
=
\Pr\!\left[\mathrm{Bin}(|\Live_\ell(S^-)|,\bar p)<k\right]
\le
\Pr[\mathrm{Bin}(n,\bar p)<k].
$$

The assumption $\bar p \le 1/2$ in the theorem ensures that $n\bar p = 2k+8\ln(W/\beta_r)$. In particular, $n\bar p\ge 2k$ and
$n\bar p\ge 8\ln(W/\beta_r)$. Hence, a Chernoff bound gives

$$
\Pr_\rho[\tau_\ell>\bar p]
\le
e^{-n\bar p/8}
\le
\frac{\beta_r}{W}.
$$
Taking a union bound over the at most $W$ rounds in $W^*$ gives

$$
\Pr_\rho[p>\bar p]
=
\Pr_\rho\!\left[\max_{\ell\in W^*}\tau_\ell>\bar p\right]
\le
\beta_r.
$$

Let $G=\{p\le\bar p\}$. On $G$, we have $p\le\bar p\le1/2$. Therefore, using $\eps_W\le1$, $e^{\eps_W}-1\le2\eps_W$, $e^{\eps_W}\le3$, and $1/(1-p)\le2$ in the preceding bound, we obtain

$$
\Pr_{S^+}[E\mid\rho]-\Pr_{S^-}[E\mid\rho]
\le
12p\,\eps_W\Pr_{S^-}[E\mid\rho]+3p\,\delta_W
\le
12\bar p\,\eps_W\Pr_{S^-}[E\mid\rho]+3\bar p\,\delta_W.
$$
Thus, using $1+x\le e^x$,

$$
\Pr_{S^+}[E\mid\rho]
\le
e^{12\bar p\eps_W}\Pr_{S^-}[E\mid\rho]
+
3\bar p\,\delta_W,
$$
Finally, averaging over $\rho$ and using $\Pr[G^c]\le\beta_r$ gives

$$
\Pr_{S^+}[E]
\le
e^{12\bar p\eps_W}\Pr_{S^-}[E]
+
3\bar p\,\delta_W+\beta_r,
$$
The reverse direction follows by the same argument after interchanging $S^+$ and $S^-$.
\end{proof}

\section{Marginal worst case negative result for CountDistinct}\label{app:reconstruction}

This appendix proves the marginal worst-case lower bound for CountDistinct stated in Theorem~\ref{thm:marginal-lb}, where ``marginal'' means that accuracy is required separately at each time step, rather than simultaneously over the entire output sequence; thus, the lower bound already holds under this weaker accuracy requirement. We first recall a foundational result of Dwork, McSherry, and Talwar \cite{DMT07} on robust decoding, in the line of reconstruction attacks initiated by Dinur and Nissim \cite{DN03}.

\paragraph{Background: Robust LP Decoding.}
Suppose we have a secret dataset $y\in\{0,1\}^n$.
The results of Dwork, McSherry, and Talwar~\cite[Theorem 23]{DMT07} imply that
there exists a universal constant $\rho^*>0$ such that, for every constant
$\gamma<\rho^*$, there exist universal constants $c_1\geq 1$ and $c_2>0$
such that the following holds. For every coefficient vector $s\in\{-1,1\}^n$, define the linear query

$$
Q_s(y):=\langle s,y\rangle=\sum_{i=1}^n s_i y_i.
$$

Let $s^{(1)},\ldots,s^{(m)}\in\{-1,1\}^n$, where $m=c_1 n$, be
sampled independently and uniformly at random. The corresponding
linear queries are $Q_{s^{(1)}},\ldots,Q_{s^{(m)}}$, and their possibly noisy answers $z_1,\ldots,z_m$ form the
answer vector $z=(z_1,\ldots,z_m)\in\mathbb{R}^m$, where $z_r$
corresponds to the query $Q_{s^{(r)}}$. With probability $1-\exp(-\Omega(n))$ over the sampled query set,
the following guarantee holds simultaneously for every
$y\in\{0,1\}^n$: given the sampled queries and any answer vector $z$
for which at least a $(1-\gamma)$ fraction of its coordinates satisfy
$\left|z_r-Q_{s^{(r)}}(y)\right|\leq \alpha$, the linear programming decoder of~\cite{DMT07}, followed by
coordinate-wise rounding, reconstructs a database that differs from
$y$ in at most $\max\{C,(C'\alpha)^2\}$ coordinates, for universal
constants $C,C'>0$. In particular, by taking
$\alpha=c_2\sqrt n$ for a sufficiently small universal constant
$c_2>0$, and fixing $\gamma>0$ below the constant corruption threshold
of~\cite{DMT07}, for all sufficiently large $n$ the decoder recovers all but $n/100$ of the bits of $y$.

Using this background, we state and prove the lemma.

\begin{lemma}[Reconstruction from marginally accurate answers]\label{lem:reconstruction}
There are constants $c_1\ge1$ and $c_2,\beta^*,\delta^*>0$ such that for all sufficiently large $n$ there exist queries $s^{(1)},\dots,s^{(k)}\in\{-1,1\}^n$ with $k=c_1\cdot n$ for which no randomized algorithm $B$ mapping $y\in\{0,1\}^n$ to $(b_1,\dots,b_k)\in\mathbb{R}^k$ can satisfy both:
\begin{enumerate}
\item[(i)] for every $y$ and every $j$, $\Pr[|b_j-\langle s^{(j)},y\rangle|\le c_2\sqrt n]\ge1-2\beta^*$; and
\item[(ii)] $B$ is $(1,\delta^*)$-differentially private with respect to changing one coordinate of $y$.
\end{enumerate}
\end{lemma}

\begin{proof}
Set $\delta^*=0.1$, and choose $\beta^*>0$ sufficiently small so that $10\beta^*<\gamma$. Suppose $B$ satisfies both conditions. We fix $s^{(1)},\ldots,s^{(k)}$ to be a sign-query set satisfying the robust reconstruction guarantee established above. Since the random construction above produces such a set with high probability, such a fixed set exists. To derive a contradiction, consider the experiment in which the secret dataset $y$ is sampled uniformly from $\{0,1\}^n$, independently of the internal randomness of $B$, and $B$ is run on this input.

Run the decoder described above on the answers $b_1,\ldots,b_k$ produced by $B$ to obtain a reconstructed dataset $\hat y\in\{0,1\}^n$. By condition (i), the probability that any single answer $b_j$ is corrupted (i.e., has error $> c_2\sqrt{n}$) is at most $2\beta^*$. The corruption events across different answers need not be independent, so the number of corrupted answers is not necessarily binomial. Nevertheless, by linearity of expectation, the expected fraction of corrupted answers, $\mathbb{E}[\#{\text{corrupted answers}}/k]$, is at most $2\beta^*$. By Markov's inequality, the probability that the actual fraction of corrupted answers exceeds $10\beta^*$ is at most $(2\beta^*) / (10\beta^*) = 1/5$. Therefore, with probability at least $4/5$, all but a $10\beta^*$ fraction of the answers are within $c_2\sqrt n$ of the truth.

By our choice of $\beta^*$, this $10\beta^*$ corrupted fraction is below the tolerance threshold $\gamma$ established above. Conditioned on the corrupted fraction being below the threshold, the decoder errs on at most $n/100$ coordinates. Hence, in this experiment, the expected fraction of correctly guessed bits is bounded from below by the probability that the decoder succeeds multiplied by its accuracy when it does succeed:

$$
\mathbb{E}\left[\frac1n|\{i:\hat y_i=y_i\}|\right]\;\ge\;\frac45\cdot\frac{99}{100}\;=\;0.792.
$$

On the other hand, let $D$ denote the decoder. Since
$\hat y=D(B(y))$ is obtained by post-processing the output of $B$,
the mapping $y\mapsto\hat y$ is also $(1,\delta^*)$-differentially
private. Therefore, differential privacy imposes an upper bound on
the probability that the reconstructed bit $\hat y_i$ correctly
guesses the original bit $y_i$. Fix an index $i$ and condition on
the other coordinates $y_{-i}$. The two values of
$y_i\in\{0,1\}$ give neighboring inputs, so, writing
$p_b=\Pr[\hat y_i=1\mid y_i=b,y_{-i}]$ and
$1-p_b=\Pr[\hat y_i=0\mid y_i=b,y_{-i}]$, the privacy
guarantee applied to the events $\{\hat y_i=1\}$ and $\{\hat y_i=0\}$
in the corresponding directions gives
$p_1\le e\cdot p_0+\delta^*$ and
$1-p_0\le e\cdot(1-p_1)+\delta^*$. Because $y$ is uniformly distributed over $\{0,1\}^n$, the bit $y_i$
is uniformly distributed even after conditioning on $y_{-i}$, and the probability of correctly guessing $y_i$ is

$$
\Pr[\hat y_i=y_i\mid y_{-i}]
=
\frac12\Pr[\hat y_i=1\mid y_i=1,y_{-i}]
+
\frac12\Pr[\hat y_i=0\mid y_i=0,y_{-i}]
=
\frac12(p_1+1-p_0).
$$

Manipulating the privacy bounds, we can upper bound $p_1 - p_0$:
\begin{align*}
p_1 + 1 - p_0 &\le e\cdot p_0 + e(1-p_1) + 2\delta^* \\
(p_1 - p_0) + 1 &\le e(1 - (p_1 - p_0)) + 2\delta^* \\
(p_1 - p_0)(e+1) &\le e - 1 + 2\delta^* \\
p_1 - p_0 &\le \frac{e-1}{e+1} + \frac{2\delta^*}{e+1}.
\end{align*}
Plugging this into our success probability yields:

$$
\Pr[\hat y_i=y_i]\;=\;\frac12(p_1-p_0+1)\;\le\;\frac12\left(\frac{e-1}{e+1} + \frac{2\delta^*}{e+1} + 1\right) \;=\; \frac{e}{e+1} + \frac{\delta^*}{e+1}.
$$

Since $\frac{e}{e+1} \approx 0.73105 \le 0.732$ and $\frac{1}{e+1} \le \frac{1}{2}$, we obtain $\Pr[\hat y_i=y_i] \le 0.732 + \frac{\delta^*}{2}$. Averaging over all coordinates $i$, the expected fraction of correctly guessed bits permitted by differential privacy is at most $0.732 + \delta^*/2$.

However, since $\delta^*=0.1$, the maximum accuracy permitted by differential privacy ($< 0.792$) strictly contradicts the accuracy achieved by the LP decoder ($\ge 0.792$). This contradiction proves that no such algorithm $B$ can exist.
\end{proof}

We now prove the marginal worst case bound.

\begingroup
\def\thetheorem{\ref{thm:marginal-lb}}
\begin{theorem}[Marginal worst case bound, restated]
There are constants $c,\beta^*,\delta^*>0$ such that for all sufficiently large $T$, no dynamic algorithm for $\mathrm{CountDistinct}$ that is $(1,\delta^*)$-private is $(c\cdot T^{1/4},\beta^*)$-accurate per round on the class of all input sequences.
\end{theorem}
\addtocounter{theorem}{-1}
\endgroup

\begin{proof}
Suppose $\A$ is $(1,\delta^*)$-private and
$(\alpha,\beta^*)$-accurate per round on all sequences. Recall that
$\mathrm{CD}(D):=|\{x\in\X:\exists t\text{ such that }(t,x)\in D\}|$
denotes the exact CountDistinct value on a dataset $D$, whereas
$a_\ell$ denotes the possibly noisy answer released by $\A$ at round
$\ell$. Let $n$ be a
parameter, let $k=c_1n$, and fix the queries
$s^{(1)},\dots,s^{(k)}$ of Lemma~\ref{lem:reconstruction}. For
$y\in\{0,1\}^n$, let

$$
Y:=\{i\in[n]:y_i=1\},
\qquad
Q_j:=\{i\in[n]:s_i^{(j)}=1\}
\quad\text{for every }j\in[k].
$$

Define the sequence $S(y)$ over universe $[n]$ as follows. At round
$i\in[n]$, add value $i$ if $y_i=1$, and otherwise leave the round
empty. Then, for every $j=1,\dots,k$, use a block of $2n$ rounds.
During the first $n$ rounds of the block, add one copy of every value
in $Q_j$, using empty rounds as needed. During the next $n$ rounds,
delete exactly those copies, again using empty rounds as needed, with
each deletion naming the stamp of the copy added in that block. Pad with empty rounds to horizon $T$. The construction fits whenever $(2k+1)n\leq T$, so we may take $n=\Theta(\sqrt{T})$.

Two observations about this family are important. First, at round $n$, the set of values present is exactly $Y$, and hence 
$$\mathrm{CD}(\Data_n)=|Y|=\|y\|_0.$$
Let $m_j$ denote the last round of the first half of block $j$, after all additions of that block have been processed and before any of its deletions occur. At round $m_j$, the set of values present is exactly $Y\cup Q_j$. Therefore,

$$
\mathrm{CD}(\Data_{m_j})=|Y\cup Q_j|.
$$

Moreover,

$$
\begin{aligned}
\langle s^{(j)},y\rangle
&=
|Y\cap Q_j|-|Y\setminus Q_j|\\
&=
2|Y\cap Q_j|-|Y|\\
&=
|Y|+2|Q_j|-2|Y\cup Q_j|\\
&=
\mathrm{CD}(\Data_n)+2|Q_j|
-2\mathrm{CD}(\Data_{m_j}).
\end{aligned}
$$

Consequently, define $b_j:=a_n+2|Q_j|-2a_{m_j}$. The difference between $b_j$ and the true answer to the query is

$$
\begin{aligned}
b_j-\langle s^{(j)},y\rangle
&=
\bigl(a_n-\mathrm{CD}(\Data_n)\bigr)
-
2\bigl(a_{m_j}-\mathrm{CD}(\Data_{m_j})\bigr).
\end{aligned}
$$

By the per-round accuracy of $\A$, each of the two errors on the
right-hand side has absolute value at most $\alpha$, except with
probability $\beta^*$. Therefore, by a union bound, with probability
at least $1-2\beta^*$ both bounds hold, and in this event

$$
\left|b_j-\langle s^{(j)},y\rangle\right|
\leq
\left|a_n-\mathrm{CD}(\Data_n)\right|
+
2\left|a_{m_j}-\mathrm{CD}(\Data_{m_j})\right| \leq 3\alpha.
$$

Second, changing one coordinate $y_i$ changes only the update at
round $i$: one sequence contains the addition of value $i$, whereas
the other leaves that round empty. All updates in the query blocks
are identical in the two sequences, and every block deletion names
only a copy added within that same block. Hence, the two sequences
differ by exactly one addition and are neighboring according to
Definition~\ref{def:neighboring}. Since
$B(y)=(b_1,\dots,b_k)$ is obtained by post-processing the transcript
of $\A$ on $S(y)$, the mapping $B$ is
$(1,\delta^*)$-differentially private.

The two observations show that $B(y)=(b_1,\dots,b_k)$ is
$(1,\delta^*)$-differentially private and that, for every $y$ and every
$j$,

$$
\Pr\left[
\left|b_j-\langle s^{(j)},y\rangle\right|\leq 3\alpha
\right]\geq 1-2\beta^*.
$$

If $3\alpha\leq c_2\sqrt{n}$, then $B$ satisfies both conditions of
Lemma~\ref{lem:reconstruction}, which is impossible. Hence,

$$
\alpha>\frac{c_2}{3}\sqrt{n}=\Omega(T^{1/4}),
$$

where the last equality follows from $n=\Theta(\sqrt{T})$. Choosing the
constant $c>0$ in the theorem statement sufficiently small gives the
claimed result.
\end{proof}

\end{document}